\documentclass[11pt,a4paper,reqno]{amsart}%
\usepackage{amsthm,amsmath,amsfonts,amssymb,amsxtra,appendix,bookmark,dsfont,latexsym}
\usepackage{amsfonts,mathrsfs}

\makeatletter
\usepackage{hyperref}
\usepackage{delarray,a4,color}
\usepackage{euscript}
\usepackage[latin1]{inputenc}

\theoremstyle{plain}
\newtheorem{theorem}{Theorem}[section]
\newtheorem{lemma}[theorem]{Lemma}

\newtheorem{proposition}[theorem]{Proposition}

\theoremstyle{remark}
\newtheorem{remark}[theorem]{Remark}

\numberwithin{equation}{section}

\DeclareMathOperator{\Tr}{Tr}
\DeclareMathOperator{\tr}{Tr}

\def\geqslant{\ge}
\def\leqslant{\le}
\def\bq{\begin{eqnarray}}
\def\eq{\end{eqnarray}}
\def\bqq{\begin{eqnarray*}}
\def\eqq{\end{eqnarray*}}

\def\eps{\varepsilon}
\def\wto{\rightharpoonup}

\newcommand{\norm}[1]{\left\lVert #1 \right\rVert}

\newcommand\1{{\ensuremath {\mathds 1} }}

\renewcommand{\epsilon}{\varepsilon}

\def\cF {\mathcal{F}}

\def\R {\mathbb{R}}
\def\NN {\mathbb{N}}
\def\C {\mathbb{C}}

\def\cG {\mathcal{G}}
\def\cP {\mathcal{P}}

\def\cE {\mathcal{E}}

\def\R {\mathbb{R}}
\def\C {\mathbb{C}}

\def\gH{\mathfrak{H}}

\newcommand\pscal[1]{{\ensuremath{\left\langle #1 \right\rangle}}}

\renewcommand{\leq}{\leqslant}
\renewcommand{\geq}{\geqslant}

\newcommand{\bA}{\mathbf{A}}
\newcommand{\bJ}{\mathbf{J}}

\newcommand{\cR}{\mathcal{R}}

\newcommand{\nablap}{\nabla^{\perp}}
\newcommand{\EAF}{E ^{\mathrm{af}}}
\newcommand{\cEAF}{\cE ^{\mathrm{af}}}
\newcommand{\uAF}{u ^{\mathrm{af}}}
\newcommand{\curl}{\mathrm{curl}}
\newcommand{\tH}{\tilde{H}}
\newcommand{\sym}{\mathrm{sym}}
\newcommand{\loc}{\mathrm{loc}}
\newcommand{\cDAF}{\mathscr{D}^{\mathrm{af}}}
\newcommand{\cMAF}{\mathcal{M}^{\mathrm{af}}}
\newcommand{\HAF}{H^{\mathrm{af}}}

\title[Average field approximation]{The average field approximation for almost bosonic extended anyons} 

\author[D. Lundholm]{Douglas LUNDHOLM}
\address{KTH Royal Institute of Technology, Department of Mathematics, SE-100 44 Stockholm, Sweden}
\email{dogge@math.kth.se}

\author[N. Rougerie]{Nicolas ROUGERIE}
\address{CNRS \& Universit\'e Grenoble Alpes, LPMMC (UMR 5493), B.P. 166, F-38042 Grenoble, France}
\email{nicolas.rougerie@lpmmc.cnrs.fr}

\date{September, 2015}

\begin{document}

\begin{abstract}
Anyons are 2D or 1D quantum particles with intermediate statistics, 
interpolating between bosons and fermions. We study the ground state of a 
large number $N$ of 2D anyons, in a scaling limit where the statistics 
parameter $\alpha$ is proportional to $N ^{-1}$ when $N\to \infty$. 
This means that the statistics is seen as a ``perturbation from the bosonic 
end''. 
We model this situation in the magnetic gauge picture by bosons interacting 
through long-range magnetic potentials. 
We assume that these effective statistical gauge potentials are generated by 
magnetic charges carried by each particle, smeared over discs of radius $R$ 
(extended anyons). 
Our method allows to take $R\to 0$ not too fast at the same time as 
$N\to \infty$. In this limit we rigorously justify the so-called 
``average field approximation'': the particles behave like independent, identically distributed 
bosons interacting via a self-consistent magnetic field.  
\end{abstract}

\maketitle

\setcounter{tocdepth}{2}
\tableofcontents

\section{Introduction}\label{sec:intro}

In lower dimensions there are possibilities for quantum statistics different 
from bosons and fermions, so called intermediate or fractional statistics. 
Due to the prospect that such particles, termed anyons (as in \emph{any}thing 
in between bosons and fermions), could arise as effective quasiparticles in 
many-body quantum systems confined to lower dimensions, 
there has been a great interest over the last three decades in figuring 
out the behavior of such statistics 
(see \cite{Froehlich-90,IenLec-92,Khare-05,Lerda-92,Myrheim-99,Ouvry-07,Wilczek-90}
for extensive reviews). 
In one dimension, one can view the Lieb-Liniger model \cite{LieLin-63} 
as providing an example of effective interpolating statistics.
Although initially regarded as purely hypothetical 
while at the same time offering substantial analytical
insight thanks to its exact solvability,
this system has now been realized concretely in the laboratory 
\cite{KinWenWei-04,Paredes_et_al-04}.
Much less is known concerning fractional statistics 
in the two-dimensional setting, 
conjectured~\cite{AroSchWil-84} to be relevant 
for the fractional quantum Hall effect 
(see~\cite{Goerbig-09,Laughlin-99} for review), 
and it is indeed a very challenging theoretical 
question to figure out even the ground state properties of an ideal 2D 
many-anyon gas, parameterized by a single statistics phase $e^{i\pi\alpha}$,
or a periodic real parameter $\alpha$ (with $\alpha=0$ corresponding to 
bosons and $\alpha=1$ to fermions). 
On the rigorous analytical side,
some recent progress in this 
direction has been achieved in \cite{LunSol-13a,LunSol-13b,LunSol-14} 
where a better understanding of the 
ground state energy was obtained in the case that 
$\alpha$ is an odd numerator fraction.
Numerous approximative descriptions have also been proposed over the years, 
such as e.g. in~\cite{ChiSen-92}
where the problem was approached from both the bosonic and the fermionic ends, 
with a harmonic trapping potential.
Here, equipped with new methods in many-body spectral theory,
we will re-visit 
this question from the perspective of a perturbation 
around bosons, i.e. in a regime where $\alpha$ is small.

\subsection{The model}

One may formally think of 2D anyons with statistics parameter $\alpha$ as 
being described by an $N$-body wave function of the form
$$ \Psi (x_1,\ldots,x_N) = \prod_{j<k} e ^{i\alpha \phi_{jk}}\tilde{\Psi} (x_1,\ldots,x_N),\quad \phi_{jk} = \mathrm{arg} \frac{x_j-x_k}{|x_j-x_k|}, $$
where $\tilde{\Psi}$ is a bosonic wave function, i.e. symmetric under 
particle exchange. This way 
$$ \Psi (x_1,\ldots,x_j,\ldots,x_k,\ldots,x_N) = e^{i\alpha\pi} \Psi (x_1,\ldots,x_k,\ldots,x_j,\ldots,x_N)$$
and the behavior under particle exchange interpolates betweens bosons 
($\alpha = 0$) and fermions ($\alpha = 1$). 
Of course the wave function is then in general
not single-valued, and this description 
is not easy to use in practice. 
To describe free anyons, one way out is to realize that acting on $\Psi$ 
with the free Hamiltonian
$$ \sum_{j=1}^N \left( -\Delta_j + V(x_j) \right) $$
is equivalent to acting on the bosonic wave function $\tilde{\Psi}$ with an 
effective, $\alpha$-dependent Hamiltonian.

The Hamiltonian obtained in this way for $N$ identical and ideal 2D anyons in 
a trapping potential 
$V:\R ^2 \to \R ^+$ 
reads (see~\cite{Froehlich-90,IenLec-92,Khare-05,Lerda-92,Myrheim-99,Ouvry-07,Wilczek-90} 
for review 
and~\cite{LunSol-13a,LunSol-13b,LunSol-14} for recent mathematical studies)
\begin{equation}\label{eq:pre hamil}
	H_N := \sum_{j=1}^N \left( (p_j + \alpha\bA_j)^2 + V(x_j) \right)
\end{equation}
where 
$$ p_j = -i \nabla_j$$
is the usual momentum operator for particle $j$ and 
(denoting $(x,y)^\perp := (-y,x)$)
\begin{equation}\label{eq:pre potential}
	\bA_j := \sum_{k\neq j} \frac{(x_j-x_k) ^{\perp}}{|x_j - x_k| ^2}  
\end{equation}
is the (normalized\footnote{For increased clarity we will in this work
separate $\alpha$ from $\bA$, so that $\bA$ corresponds to the statistical 
vector potential of fermions modeled as bosons.})
statistical gauge vector potential felt by particle $j$ due to the 
influence of all other particles. 
The statistics parameter is denoted by $\alpha$, 
corresponding to a statistical phase $e^{i\alpha\pi}$ under a continuous
simple interchange of two particles. 
In this so-called ``magnetic gauge picture'', 2D anyons are thus described 
as bosons, each of them carrying an Aharonov-Bohm magnetic flux of strength $\alpha$.

We shall in this work assume 
\begin{equation}\label{eq:scale alpha}
	\alpha = \frac{\beta}{N-1} \to 0 \ \ \text{when} \ \ N\to \infty,
\end{equation}
where $\beta$ is a given, fixed constant. 
We consider~\eqref{eq:pre hamil} as (formally) acting on the bosonic 
$N$-particle space $L_\sym^2 (\R ^{2N})$,
which together with the condition~\eqref{eq:scale alpha} means that we consider 
almost bosonic anyons. Note that if we simply took $\alpha \to 0$ at fixed $N$, 
we would recover ordinary bosons at leading order. 
One could then only see the effect of the non-trivial statistics in a 
perturbative expansion, a route followed e.g. 
in~\cite{Sen-91,ComCabOuv-91,SenChi-92,Ouvry-94,ComMasOuv-95}. 
However, if $N\to \infty$ with fixed $\beta$ as above, 
the anyon statistics has a leading order effect, 
manifest through a particular mean-field model with a 
self-consistent magnetic field of strength $\sim \beta$, studied e.g. 
in~\cite{ChiSen-92,Trugenberger-92b,Trugenberger-92,IenLec-92,Westerberg-93}
and often called \emph{the average field approximation}
(see \cite{Wilczek-90,IenLec-92} for review). 
Our aim in this work is to justify this description rigorously.

\medskip

Actually, the Hamiltonian~\eqref{eq:pre hamil} is too singular to be 
considered as acting on a pure tensor product 
$u^{\otimes N} \in \bigotimes_\sym^N \gH$, 
however regular the function $u$ of the one-particle space 
$\gH \subseteq L ^2 (\R ^2)$. 
We refer to~\cite[Section 2.1]{LunSol-14} for a discussion of the domain 
of~\eqref{eq:pre hamil},
which requires the removal of the two-particle diagonals from the 
configuration space $\R^{2N}$.
One way to circumvent this issue is to reintroduce a length scale $R$ over 
which the magnetic charge is smeared. This so-called ``extended anyons'' 
model is discussed in~\cite{Mashkevich-96,Trugenberger-92b,ChoLeeLee-92}, 
and is sometimes argued to be the correct physical description
for anyons arising as quasi-particles in condensed-matter systems. 
In this paper we will allow $R$ to become small when $N\to \infty$, in which 
case we recover the point-like anyons point of view, at least if one is 
willing to ignore the issue of non-commuting limits.

Let us consider the 2D Coulomb potential generated by a unit charge smeared 
over a disc of radius $R$:
\begin{equation}\label{eq:smeared Coul}
w_R (x) := \log |\:.\:| \ast \frac{\1_{B(0,R)}}{\pi R^2} (x). 
\end{equation}
Observing that (with the convention $w_0 := \log |\:.\:|$)
$$
	\nabla ^{\perp} w_0 (x) = \frac{x ^{\perp}}{|x| ^2},
	\quad \text{with} \quad
	B_0(x) = \nablap \cdot \nablap w_0 = \Delta w_0 = 2\pi \delta_0,
$$
the natural regularization of $\bA_j$ corresponding to the extended-anyons 
model is given by
\begin{equation}\label{eq:potential}
	\bA ^R_j := \sum_{k\neq j} \nablap w_R (x_j-x_k), 
\end{equation}
leading to the regularized Hamiltonian
\begin{equation}\label{eq:hamil}
	H_N ^R := \sum_{j=1}^N \left( (p_j + \alpha \bA_j^R) ^2 + V(x_j) \right).
\end{equation}
We shall denote 
\begin{equation}\label{eq:gse}
	E ^R (N) := \inf \sigma (H_N ^R)
\end{equation}
the associated ground state energy (lowest eigenvalue)
for $N$ extended anyons.

For fixed $R>0$ this operator is self-adjoint on $L_\sym^2 (\R ^{2N})$ and one 
can even expand the squares to obtain a sum of terms that are all 
symmetric and relatively form-bounded with respect to the 
$\alpha=0$ non-interacting operator\footnote{By
the boundedness of $\nabla w_R$ and using Cauchy-Schwarz, all terms are 
infinitesimally form-bounded in terms of $H_N(\alpha=0)$
and hence $H_N^R$ is a uniquely defined self-adjoint operator
by the KLMN theorem~\cite[Theorem X.17]{ReeSim2}.
We shall assume $V$ is such that a form core is given by $C_c^\infty(\R^2)$.}.
This gives 
\begin{align}\label{eq:expand hamil}
H_N ^R &= \sum_{j=1}^N \left( p_j ^2 + V(x_j) \right) \nonumber \\
&+ \alpha \sum_{j\neq k} \left( p_j \cdot \nablap w_R (x_j - x_k ) + \nablap w_R (x_j - x_k ) \cdot p_j \right) \nonumber \\
&+ \alpha ^2 \sum_{j\neq k \neq \ell} \nablap w_R (x_j-x_k) \cdot \nablap w_R (x_j-x_\ell) \nonumber \\
&+ \alpha ^2 \sum_{j\neq k} |\nabla w_R (x_j-x_k)| ^2.
\end{align}
We also note that by the diamagnetic inequality 
(see, e.g., \cite[Theorem 7.21]{LieLos-01} for $R>0$,
and \cite[Lemma 4]{LunSol-14} for $R=0$)
$$
	\langle \Psi, H_N^R \Psi \rangle 
	\ge \langle |\Psi|, H_N(\alpha=0) |\Psi|\rangle,
$$
and hence $E^R(N) \ge N E_0$ for arbitrary $\alpha$,
with $E_0$ the ground state energy of the one-body operator
$H_1 = p^2 + V$.

\subsection{Average field approximation}

The few-anyons problem can be studied within perturbation theory, yielding 
satisfactory information on the ground state and low-lying excitation 
spectrum~\cite{Sen-91,ComCabOuv-91,Ouvry-94,ComMasOuv-95}. 
For many anyons however, it is hard to obtain results this way. 
A possible approximation to obtain a more tractable model when $N$ is large 
consists in seeing the potential~\eqref{eq:pre potential} 
or~\eqref{eq:potential} as being independent of the precise positions
$x_j$ and instead generated by the mean distribution of the particles 
(whence the name, average field approximation \cite{Wilczek-90})
\begin{align}\label{eq:avg field}
	\bA [\rho] &:= \nablap w_0 \ast \rho, \nonumber\\
	\bA^R [\rho] &:= \nablap w_R \ast \rho,
\end{align}
where $\rho$ is the one-body density (normalized in $L^1(\R ^2)$) of a given 
bosonic wave function $\Psi$
$$ 
	\rho(x) = \int_{\R ^{2(N-1)}} |\Psi (x,x_2,\ldots,x_N)|^2 \,dx_2\ldots dx_N,
$$
say the ground-state wave function.
One then obtains from \eqref{eq:hamil} the approximate $N$-body 
Hamiltonian
$$ 
	\HAF_N[\rho] := \sum_{j=1}^N \left( \left( p_j + N\alpha\bA ^R [\rho] \right) ^2 + V(x_j) \right).
$$
If one considers $\rho$ as fixed, the ground state of this 
non-interacting magnetic Hamiltonian acting 
on $L ^2_\sym (\R ^{2N})$ is a pure Bose condensate 
$$ \Psi_N = u ^{\otimes N}, $$
where $u \in L^2(\R^2)$ should minimize
$$ 
	\langle u \big| (p+ N\alpha \bA^R[\rho] ) ^2 + V \big| u \rangle  
	= N^{-1} \left\langle \Psi_N, \HAF_N[\rho] \Psi_N \right\rangle.
$$
For consistency, one should then impose that 
$$ |u| ^2 = \rho, $$
which leads to a non-linear minimization problem. 
One thus looks for the minimum $\EAF$ and minimizer $\uAF$ of 
the following \emph{average-field energy functional}
(recall the notation $\beta \sim N \alpha$)
\begin{equation}\label{eq:avg func}
	\cEAF_{R} [u] := \int_{\R ^2} \left( \left| \left( \nabla 
		+ i \beta \bA^{R}[|u|^2] \right) u \right|^2 + V|u|^2 \right)
\end{equation}
under the unit mass constraint
$$ \int_{\R ^2} |u| ^2 = 1.$$
Note that, for this problem to be independent of $N$ it is pretty natural to
--- in line with~\eqref{eq:scale alpha} ---
assume that $\beta \sim N\alpha$ is fixed.
It is not difficult to see that if $N\alpha \to 0$ we recover at leading order 
a non-interacting theory, and we are back to the usual perturbation scheme. 
We should point out that the limiting
functional $\cEAF_{R=0}$, which defines a strictly two-dimensional model 
of particles with a self-generated magnetic field 
$B(x) = \curl \,\beta \bA[\rho] (x) = 2\pi\beta\rho(x)$
without propagating degrees of freedom
and to which one could further consider adding an external magnetic field,
is also of relevance for various Chern-Simons formulations of anyonic theories
(see, e.g., \cite{ZhaHanKiv-89,Zhang-92,IenLec-92,CheSmi-14}).

\subsection{Average field versus mean field}

In principle, the average field approximation does not require that the true 
ground state of $H_N ^R$ be Bose-condensed.
In fact, the most common application of it has been in perturbing 
around fermions $\alpha=1$ 
\cite{FetHanLau-89,IenLec-92,Trugenberger-92b,Trugenberger-92,Westerberg-93}
(this has even been argued to be preferable \cite{ChenWilWitHal-89,Wilczek-90}),
and usually one even restricts to the homogeneous setting with $\rho$ a constant.
However, the case of fixed $\beta \sim N \alpha$ which is natural for the study 
of~\eqref{eq:avg func}, places the limit $N\to \infty$ of the original many-body 
problem in a mean-field-like regime for bosons. 
Indeed, observe that in~\eqref{eq:expand hamil}, the two-body terms in the 
second line and the three-body term in the third line weigh a total $O(N)$ in 
the energy in this regime, comparable to the one-body term in the first line. 
The two-body term in the fourth line is of much smaller order, $O(1)$ roughly, 
which is fortunate because of its singularity. Actually, if one takes bluntly 
$R=0$, the potential $|\nabla w_0| ^2$ appearing in this term is not locally 
integrable, and hence an ansatz $\Psi_N = u ^{\otimes N}$ would lead to an 
infinite energy. For extended anyons, $R>0$ and this term can be safely dropped 
for leading order considerations. 

The study of the regime~\eqref{eq:scale alpha} thus resembles a lot the usual 
mean-field limit for a large bosonic system 
(see~\cite{LewNamRou-13,LewNamRou-14,Rougerie-cdf} and references therein), but 
with important differences:
\begin{itemize}
\item The effective interaction is peculiar: it comprises a three-body term, 
	and a two-body term which mixes position and momentum variables.
\item The limit problem~\eqref{eq:avg func} comprises an effective 
	self-consistent \emph{magnetic} field. 
	A term in the form of a self-consistent \emph{electric} field is more usual.
\item One should deal with the limit $R\to 0$ at the same time as 
	$N \to \infty$, which is reminiscent of the NLS and GP limits for trapped 
	Bose gases~\cite{LieSeiYng-00,LieSeiSolYng-05,LieSei-06,LewNamRou-14,NamRouSei-15}. 
\end{itemize}

In order to make the analogy more transparent, we rewrite, 
for any normalized $N$-body bosonic wave function 
$\Psi_N \in L^2_\sym(\R^{2N})$
\begin{align}\label{eq:ener dens mat}
N^{-1} \left\langle \Psi_N \big| H_N ^R \big| \Psi_N \right\rangle &= \tr \left[ (p ^2 + V) \gamma_N ^{(1)} \right] \nonumber \\
&+ \beta \tr\left[ \left( p_1 \cdot \nablap w_R (x_1-x_2) + \nablap w_R (x_1-x_2) \cdot p_1\right) \gamma_N ^{(2)} \right]\nonumber\\
&+ \beta ^2 \frac{N-2}{N-1} \tr\left[ \left(\nablap w_R (x_1-x_2) \cdot \nablap w_R (x_1-x_3) \right) \gamma_N ^{(3)} \right]\nonumber\\ 
&+ \beta ^2 \frac{1}{N-1} \tr\left[ |\nabla w_R (x_1-x_2)| ^2 \gamma_N ^{(2)} \right],
\end{align}
where 
$$\gamma_N ^{(k)} := \tr_{k+1 \to N} \left[ |\Psi_N\rangle \langle \Psi_N|\right]$$ 
is the $k$-body density matrix of the state $|\Psi_N \rangle \langle \Psi_N |$, 
normalized to have trace $1$. 
The notation here means that we trace out the last $N-k$ variables from the 
integral kernel of $|\Psi_N \rangle \langle \Psi_N |$.

Since all terms at least at first sight weigh $O(1)$ or less, the folklore suggests to use an ansatz 
\begin{align}\label{eq:MF ansatz}
	\Psi_N &= u ^{\otimes N}  \nonumber\\ 
	\gamma_N ^{(k)} &= | u ^{\otimes k} \rangle \langle u ^{\otimes k}|.
\end{align}
Inserting this in the energy, dropping the last term, which is of order 
$N^{-1}$ at least for fixed~$R$, we obtain to leading order
\begin{equation}\label{eq:MF energy}
	N^{-1} \left\langle \Psi_N \big| H_N ^R \big| \Psi_N \right\rangle 
	\approx \cEAF_R [u].
\end{equation}
Indeed, on the one hand, 
\begin{align}\label{eq:two body term}
&\tr\left[ \left( p_1 \cdot \nablap w_R (x_1-x_2) + \nablap w_R (x_1-x_2) \cdot p_1\right) | u ^{\otimes 2} \rangle \langle u ^{\otimes 2}| \right] \nonumber\\ 
&= i \iint_{\R ^2 \times \R ^2}  \nabla \overline{u} (x) \overline{u} (y) \cdot \nablap w_R (x - y) u(x)u(y) \,dx dy  \nonumber\\
&- i \iint_{\R ^2 \times \R ^2}   \overline{u} (x) \overline{u} (y) \nablap w_R (x - y) \cdot \nabla u(x)u(y) \,dx dy \nonumber \\
&= 2 \int_{\R^2} \bA^R[|u|^2] \cdot \bJ[u],
\end{align}
using the definition~\eqref{eq:avg field} and denoting $\bJ [u]$ the current 
\begin{equation}\label{eq:current}
\bJ [u] := \frac{i}{2} \left( u \nabla \overline{u} - \overline{u} \nabla u\right).  
\end{equation}
Note that this is really a phase current density:
$$
\bJ [u] = \rho \nabla \varphi
\quad \text{if} \quad u = \sqrt{\rho} e ^{i\varphi}.
$$
On the other hand
\begin{align}\label{eq:three body term}
&\tr\left[ \left(\nablap w_R (x_1-x_2) \cdot \nablap w_R (x_1-x_3) \right) \gamma_N ^{(3)} \right]\nonumber \\ 
&= \iiint_{\R ^2 \times \R ^2 \times \R ^2}  |u (x)| ^2 | u (y)| ^2 |u(z)| ^2 \nablap w_R (x - y) \cdot \nablap w_R (x-z) \,dxdydz \nonumber\\
&= \int_{\R ^2} |u| ^2 \left| \bA^R[|u|^2] \right| ^2,
\end{align}
and it suffices to combine these identities in \eqref{eq:ener dens mat}
(and approximate $(N-2)/(N-1)\sim 1$) to obtain the desired expression 
\eqref{eq:MF energy} for the energy.

\subsection{Main results}

We may now state our main theorem, justifying the average field approximation 
in the almost-bosonic limit at the level of the ground state. 
For technical reasons we assume that the one-body potential is confining 
\begin{equation}\label{eq:trap pot}
		V(x) \ge c |x|^s - C, \quad s>0,
\end{equation}
and that the size $R$ of the extended anyons does not go to zero too fast in 
the limit $N\to \infty$. The rate we may handle depends on $s$. 
These assumptions are probably too restrictive from a physical point of view 
but our method of proof does not allow to relax them at present. 
Here and in the sequel, $\cEAF$ denotes the average-field 
functional~\eqref{eq:avg func} for $R=0$, and $\EAF$ its infimum under a unit 
mass constraint. 
Although we do not state it explicitly, we could also keep $R$ fixed when 
$N\to \infty$ and obtain the limit functional with finite $R$. 
The case of anyons in a bounded domain is also covered by our approach 
(modulo the discussion of boundary conditions) and the results in this case 
can be obtained by formally setting $s=\infty$ in the following.

\begin{theorem}[\textbf{Validity of the average field approximation}]\label{thm:main ener}\mbox{}\\
We consider $N$ extended anyons of radius
	$R \sim N ^{-\eta}$ in an external potential $V$ satisfying~\eqref{eq:trap pot}. We assume the relation 
	\begin{equation}\label{eq:eta restriction}
		0 < \eta < \eta_0 (s) := \frac{1}{4}\left( 1+ \frac{1}{s}\right)^{-1},
	\end{equation}
	and that the statistics parameter scales as  
	$$\alpha = \beta/(N-1)$$ 
	for fixed $\beta \in \R$.
	Then, in the limit $N\to \infty$ we have 
	for the ground-state energy
	\begin{equation}\label{eq:main ener}
		\frac{E^R(N)}{N} \to \EAF.
	\end{equation}
	Moreover, if $\Psi_N$ is a sequence of ground states for $H_N ^R$, with 
	associated reduced density matrices $\gamma_N ^{(k)}$, then 
	modulo restricting to a subsequence we have
	\begin{equation}\label{eq:main state}
		\gamma_N ^{(k)} \to \int_{\cMAF} |u ^{\otimes k} \rangle \langle u ^{\otimes k} | \,d\mu(u) 
	\end{equation}
	strongly in the trace-class when $N\to \infty$, where $\mu$ is a Borel 
	probability measure supported on the set of minimizers of $\cEAF$,
	$$
		\cMAF := \lbrace u\in L ^2 (\R ^2) : \norm{u}_{L^2} = 1, \: \cEAF [u] = \EAF \rbrace.
	$$
\end{theorem}

\bigskip 

\begin{remark}[The size of the smearing radius]\label{rem:R}\mbox{}\\
As we said before, taking $R$ not too small in the limit $N\to \infty$ is a 
requirement in our method of proof. It is likely that localization arguments 
could allow to take an $s$-independent $\eta < \eta_0 (\infty) = 1/4$, 
corresponding to the rate we obtain for anyons in a bounded domain. 
To obtain an even better rate would require important new ideas.

It is in fact not clear whether some lower bound on $R$ is a necessary 
condition for the average field description to be correct. 
For very small or zero $R$, it is still conceivable that a description in 
terms of a functional of the form of~\eqref{eq:avg func} is correct in the limit.
Indeed, our above restriction stems from the method used to 
bound the ground-state energy from below,
while for an upper bound the conditions 
on $R$ (and even the finiteness of $\beta$) can be relaxed significantly.
A further possibility would be to take short-range correlations into account 
via Jastrow factors, as in the GP limit for the usual Bose 
gas~\cite{LieSeiYng-00,LieSeiSolYng-05,LieSei-06,LewNamRou-14,NamRouSei-15}.    

It would in any case be desirable to be able to take $R\ll N ^{-1/2}$, 
the typical interparticle distance in our setting, because one could then 
argue that smearing the magnetic charges has very little effect. 
Even in the formal case $s=\infty$, Theorem~\ref{thm:main ener} requires 
$R \gg N ^{-1/4} \gg N ^{-1/2}$, which is rather stringent. 
We nevertheless obtain the functional for point-like anyons in the limit.

It is sometimes argued in the literature~\cite{Mashkevich-96,Trugenberger-92b,ChoLeeLee-92} 
that for anyons arising as quasi-particles in condensed matter physics, 
the magnetic charges \emph{should} be smeared over some length scale $R$. 
The relevant relation between $R$ and $N$ then depends on the context. 
\hfill $\Diamond$ 
\end{remark}

\bigskip

The rest of the paper contains the proof of Theorem~\ref{thm:main ener}. 
We start by collecting in Section~\ref{sec:bounds} some operator bounds on the 
different terms of the $N$-body functional. 
This is required in order to have a correct control of the terms 
as a function of the kinetic energy in the limit $R\to 0$. 
For these estimates to be of use in the large-$N$ limit 
we need an a priori bound on the kinetic 
energy of ground states of the $N$-body problem, also derived in 
Section~\ref{sec:bounds}. 
We deal with the mean-field limit in Section~\ref{sec:MF lim}, using the 
method of~\cite{LewNamRou-14}. 
Some important adaptations are required to deal with the anyonic Hamiltonian, 
and we focus on these. 
The goal here is to justify (with quantitative error bounds) the sensibility 
of the ansatz $\Psi_N = u ^{\otimes N}$ when $N$ becomes large, thus obtaining 
$\EAF_R$ as an approximation of the ground state energy per particle. 
The basic properties of the average-field functional \eqref{eq:avg func} are 
worked out in Appendix~\ref{app:AFF}. 
In particular we study the limit $R\to 0$ to finally obtain $\EAF$ as an 
approximation of the many-body ground state energy per particle.

\bigskip

\noindent\textbf{Acknowledgments.} 
We thank Michele Correggi for discussions.
Part of this work has been carried out during visits at the 
\emph{Institut Henri Poincar\'e} (Paris) and the 
\emph{Institut Mittag-Leffler} (Stockholm). 
D.L. would also like to thank LPMMC Grenoble for kind hospitality.
We acknowledge financial support from the \emph{French ANR} 
(Project Mathosaq ANR-13-JS01-0005-01),
as well as 
the grant KAW 2010.0063 from the \emph{Knut and Alice Wallenberg Foundation}
and the \emph{Swedish Research Council} grant no. 2013-4734.

\section{The extended anyon Hamiltonian}\label{sec:bounds}

In this section we give some bounds allowing to properly define and control the 
Hamiltonian~\eqref{eq:hamil}.
As previously mentioned, for extended anyons, it is possible to expand the 
Hamiltonian as in~\eqref{eq:expand hamil} and estimate it term by term. 
By the boundedness of the interaction it follows that
$H_N^R$ is defined uniquely as a self-adjoint operator on $L_\sym^2 (\R^{2N})$
with the same form domain as the non-interacting bosonic Hamiltonian 
$$\sum_{j=1}^N ( p_j ^2 + V(x_j) ).$$
However, in order to eventually take the limit $R \to 0$
we will need to deduce more precise bounds depending on $R$. These will be used to deal with the mean-field limit in Section~\ref{sec:MF lim}.

In the following we introduce a fixed reference length scale $R_0>0$, 
and always assume $R \ll R_0$. Future constants,
generically denoted by $C$, may implicitly depend on $R_0$.

\subsection{Operator bounds for the interaction terms}\label{sec:op bounds}

We start with some estimates on the different terms in~\eqref{eq:ener dens mat}, 
exploiting the regularizing effect of taking $R>0$. 
The following is standard:

\begin{lemma}[\textbf{The smeared Coulomb potential}]\label{lem:smear coul}\mbox{}\\
Let $w_R$ be defined as in~\eqref{eq:smeared Coul}. There is a constant $C>0$ depending only on $R_0$ such that  
\begin{equation}\label{eq:sup w_R}
	\sup_{B(0,R_0)} |w_R| \leq C + |\log R|, 
	\quad \sup_{\R^2} |\nabla w_R| \leq \frac{C}{R},
	\quad \sup_{B(0,R_0)^c} |\nabla w_R| \leq C.
\end{equation}
Moreover, for any $2 < p < \infty$,
\begin{equation}\label{eq:Lp w_R}
	\norm{\nabla w_R}_{L^p (\R^2)} \leq C_p \,R^{2/p-1}.
\end{equation}
\end{lemma}

\begin{proof}
A simple application of Newton's theorem~\cite[Theorem 9.7]{LieLos-01} yields
\begin{equation}\label{eq:smear coul exp}
	w_R (x) = \begin{cases}
          \log |x| \mbox{ if } |x| \geq R\\
          \log R  + \frac{1}{2}\left( \frac{|x| ^2}{R ^2} - 1 \right) \mbox{ if } 0 \leq |x| \leq R,
          \end{cases}
          \:\:
	\nabla w_R (x) = \begin{cases}
			x/|x|^2 \mbox{ if } |x| \geq R\\
			x/R^2 \mbox{ if } 0 \leq |x| \leq R,
          \end{cases}
\end{equation} 
and~\eqref{eq:sup w_R} clearly follows. For~\eqref{eq:Lp w_R} we compute 
\begin{align*}
	\norm{\nabla w_R}_{L^p(\R^2)}^p  
	&= 2\pi \int_{0} ^{R} \frac{r^p}{R^{2p}} rdr + 2\pi \int_{R}^{\infty} r^{-p} rdr
	\leq C_p^p \,R^{2-p},
\end{align*}
where $C_p > 0$ depends only on $p>2$.
\end{proof}

We first estimate the most singular term of the Hamiltonian, corresponding to 
the fourth line of~\eqref{eq:expand hamil}. 
Since it comes with a relative weight $O(N ^{-1})$ in the total energy, the 
following bound will be enough to discard it from leading order considerations. 

\begin{lemma}[\textbf{Singular two-body term}]\label{lem:sing two body}\mbox{}\\
We have that, as operators on $L^2(\R^4)$ or $L^2_\sym(\R^4)$, 
\begin{equation}\label{eq:sing two body}
	|\nabla w_R (x-y)| ^2 \leq C_\eps R^{-\eps} \, \left( p_x ^2 + 1\right)
\end{equation}
for any $\eps >0$.
\end{lemma}

\begin{proof}
We start with a well-known simple application of H\"older's and Sobolev's 
inequalities:  for any $W:\R ^2 \mapsto \R$ and $f\in C^{\infty}_c (\R^4)$
\begin{align}\label{eq:Sob}
	\langle f | W(x-y) | f \rangle &= \iint_{\R ^2 \times \R ^2} \overline{f(x,y)} W(x-y) f(x,y) \,dx dy \nonumber\\
	&\leq \norm{W}_{L^p} \int_{\R ^2}\left( \int_{\R ^2} |f(x,y)| ^{2q} dx \right) ^{1/q} dy \nonumber\\
	&\leq C \norm{W}_{L^p} \iint_{\R^2 \times \R ^2} \left( |\nabla_x f(x,y)|^2  + |f(x,y)| ^2 \right) dxdy \nonumber\\
	&= C \norm{W}_{L^p} \langle f | (-\Delta_x + 1)  \otimes \1 | f \rangle
\end{align}
where we may take any $p>1$, $q = \frac{p}{p-1} \in (1,+\infty)$, 
and we use that in $\R^2$, for any $1\le q<\infty$
$$
	\norm{g}_{L^{2q}}^2 \le C_q \left( \norm{\nabla g}_{L^2}^2 + \norm{g}_{L^2}^2 \right),
$$
see, e.g., \cite[Theorem 8.5 ii]{LieLos-01}. 
Next we may use \eqref{eq:Lp w_R} with $W = |\nabla w_R|^2$ 
and $p = 1 + \eps'$ to conclude
$$
	\norm{W}_{L^p} = \norm{\nabla w_R}_{L^{2p}}^2 
	\le C_{2p}^2 R^{2/p-2} \le C_\eps R^{-\eps},
$$
with a constant $C_\eps > 0$ for given $\eps>0$.
\end{proof}

We next deal with the two-body term mixing position and momentum, second line 
of~\eqref{eq:expand hamil}. This is somehow the most difficult term to handle, 
and it is crucial to observe that it acts on the current and not on the full 
momentum. We shall use three different bounds. 
In the following lemma,~\eqref{eq:mix two body bis} has a worst $R$-dependence 
but it behaves better for large momenta than~\eqref{eq:mix two body eps} 
and~\eqref{eq:mix two body}, a fact that will be useful when projecting 
the problem onto finite dimensional spaces in the next section. 
Estimate~\eqref{eq:mix two body} might seem a bit better 
than~\eqref{eq:mix two body eps}, but we will actually need a bound on the 
absolute value in the sequel, which is not provided by~\eqref{eq:mix two body}.

\begin{lemma}[\textbf{Mixed two-body term}]\label{lem:mix two body}\mbox{}\\
For $R<R_0$ small enough we have that, as operators on $L^2_\sym (\R^{4})$,
\begin{equation}\label{eq:mix two body bis}
	\left| p_x \cdot \nablap w_R (x-y) + \nablap w_R (x-y) \cdot p_x \right| 
	\leq C R ^{-1} \: | p_x |,
\end{equation}
\begin{equation}\label{eq:mix two body eps}
	\left| p_x \cdot \nablap w_R (x-y) + \nablap w_R (x-y) \cdot p_x \right| 
	\leq C_\eps R^{-\eps} \: (p_x^2 + 1),
	\mbox{ for all } \eps>0,
\end{equation}
and
\begin{equation}\label{eq:mix two body}
	\pm\left( p_x \cdot \nablap w_R (x-y) + \nablap w_R (x-y) \cdot p_x \right) 
	\leq C (1 + |\log R|) \: (p_x^2 + 1).
\end{equation}
\end{lemma}

\begin{proof}
The bounds~\eqref{eq:mix two body bis} and~\eqref{eq:mix two body eps} are based on the same basic computation.

\medskip

\noindent\textbf{Proof of~\eqref{eq:mix two body bis}}. First note that 
\begin{equation}\label{eq:commutator}
	p_x \cdot \nablap w_R (x-y) = \nablap w_R (x-y) \cdot p_x
\end{equation}
because $\nabla_x \cdot \nablap w_R(x-y) = 0$. 
We can then square the expression we want to estimate, obtaining 
$$
	\left( p_x \cdot \nablap w_R (x-y) + \nablap w_R (x-y)  \cdot p_x\right) ^2 
	= 4  p_x \cdot \nablap w_R (x-y) \nablap w_R (x-y)  \cdot p_x.
$$
Consequently, for any $f = f(x,y) \in C^{\infty}_c(\R^4)$,
\begin{multline*}
	\left| \left\langle f \big| \left( p_x \cdot \nablap w_R (x-y) + \nablap w_R (x-y)  \cdot p_x\right) ^2 \big| f \right\rangle \right| \\
	= 4\left| \iint_{\R^2 \times \R^2} \left( \nabla_x \bar{f} (x,y) \cdot \nablap w_R (x-y) \right) \left(\nabla_x f (x,y) \cdot \nablap w_R (x-y) \right)dxdy\right| \\
	\leq 4\iint_{\R^2 \times \R^2} \left| \nabla_x f (x,y)\right| ^2 \left| \nablap w_R (x-y) \right| ^2 dxdy. 
\end{multline*}
Inserting~\eqref{eq:sup w_R} we get 
$$ \left| \left\langle f \big| \left( p_x \cdot \nablap w_R (x-y) + \nablap w_R (x-y)  \cdot p_x\right) ^2 \big| f \right\rangle \right| \leq \frac{C}{R ^2} \iint_{\R^2 \times \R^2} \left| \nabla_x f \right| ^2 dxdy$$
and thus 
$$
	\left( p_x \cdot \nablap w_R (x-y) + \nablap w_R (x-y)  \cdot p_x\right) ^2 
	\leq \frac{C}{R ^2} p_x ^2.
$$
We deduce~\eqref{eq:mix two body bis} because the square root is operator 
monotone (see, e.g., \cite[Chapter~5]{Bhatia}).

\medskip

\noindent\textbf{Proof of~\eqref{eq:mix two body eps}}. 
We proceed in the same way but use Lemma~\ref{lem:sing two body} instead of just the rough bound~\eqref{eq:sup w_R}
(we denote $x=(x_1,x_2) \in \R^2$):
\begin{multline*}
\left| \left\langle f \big| \left( p_x \cdot \nablap w_R (x-y) + \nablap w_R (x-y)  \cdot p_x\right) ^2 \big| f \right\rangle \right|
\\ \leq 4 \iint_{\R^2 \times \R^2} \left| \nabla_x f (x,y)\right| ^2 \left| \nablap w_R (x-y) \right|^2 dxdy \\
	= \left\langle \partial_{x_1}f, \left| \nablap w_R (x-y) \right|^2 \partial_{x_1}f \right\rangle_{L^2(\R^4)}
	+ \left\langle \partial_{x_2}f, \left| \nablap w_R (x-y) \right|^2 \partial_{x_2}f \right\rangle_{L^2(\R^4)} \\
	\leq \frac{C_\eps}{R^\eps} \left( 
		\left\langle \partial_{x_1}f, \left(-\Delta_x + 1\right) \partial_{x_1}f \right\rangle
		+ \left\langle \partial_{x_2}f, \left(-\Delta_x + 1\right) \partial_{x_2}f \right\rangle
		\right) 
	\leq \frac{C_\eps}{R^\eps} \left\langle f, \left(-\Delta_x + 1\right)^2 f \right\rangle.
\end{multline*}
Thus 
$$
	\left( p_x \cdot \nablap w_R (x-y) + \nablap w_R (x-y)  \cdot p_x\right) ^2 
	\leq \frac{C_\eps}{R ^\eps} (p_x^2 + 1)^2
$$
for any $\eps >0$, 
and the desired bound again follows by taking the square root.

\medskip

\noindent\textbf{Proof of~\eqref{eq:mix two body}}. 
The idea is here a bit different. We pick $f \in C^{\infty}_c(\R^4;\C)$ 
and compute as in~\eqref{eq:two body term}
$$ \left\langle f \big| p_x \cdot \nablap w_R (x-y) + \nablap w_R (x-y) \cdot p_x \big| f \right\rangle = 2 \iint_{\R ^2 \times \R ^2} \nablap w_R (x-y) \cdot \bJ_x [f] \,dxdy$$
with 
$$ \bJ_x [f] = \frac{i}{2} \left( f \: \overline{\nabla_x f} - \overline{f} \nabla_x f\right).$$
We then split this according to 
a partition of unity $\chi + \eta  = 1$
where $\chi \equiv 1$ in the ball $B(0,R_0)$ and $\eta \equiv 1$ outside of the ball $B(0,2R_0)$:
\begin{multline*}
	\left\langle f \big| p_x \cdot \nablap w_R (x-y) + \nablap w_R (x-y) \cdot p_x \big| f \right\rangle = 2 \iint_{\R ^2 \times \R ^2} \nablap \left( \chi (x-y) w_R (x-y) \right) \cdot \bJ_x [f] \,dxdy \\
	+ 2 \iint_{\R ^2 \times \R ^2} \nablap \left( \eta (x-y) w_R (x-y) \right) \cdot \bJ_x [f] \,dxdy.
\end{multline*}
To control the $\chi$ term we use Stokes' formula and deduce 
$$ 
	2 \iint_{\R ^2 \times \R ^2} \nablap \left( \chi (x-y) w_R (x-y) \right) \cdot \bJ_x [f] \,dxdy 
	= - 2 \iint_{\R ^2 \times \R ^2} \chi(x-y) w_R (x-y) \curl_x \bJ_x [f] \,dxdy.
$$
It is easy to see that 
$$ \left| \curl_x \bJ_x [f] \right| \leq |\nabla_x f| ^2$$
pointwise, see e.g.~\cite[Lemma 3.4]{CorPinRouYng-12}. We thus obtain 
\begin{align*}
&\pm 2 \iint_{\R ^2 \times \R ^2} \nablap \left( \chi (x-y) w_R (x-y) \right) \cdot \bJ_x [f] \,dxdy \\ 
&\leq 2 \iint_{\R ^2 \times \R ^2} |\chi(x-y)||w_R (x-y)| \left| \nabla_x f (x,y)\right| ^2 dxdy \\
&\leq C (1+|\log R|) \iint_{\R ^2 \times \R ^2} \left| \nabla_x f (x,y)\right| ^2 dxdy
\end{align*}
in view of~\eqref{eq:sup w_R}. For the $\eta$ term we note that 
$$ \left| \bJ_x [f] \right| \leq |f| |\nabla_x f|.$$
Thus 
\begin{align*}
	&\pm 2\iint_{\R ^2 \times \R ^2} \nablap \left( \eta (x-y) w_R (x-y) \right) \cdot \bJ_x [f] dxdy \\
	&\leq \iint_{\R ^2 \times \R ^2} \left| \nablap \left( \eta (x-y) w_R (x-y) \right) \right| ^2 |f(x,y)| ^2 dxdy
		+ \iint_{\R ^2 \times \R ^2} \left| \nabla_x f (x,y)\right| ^2 dxdy \\
	&\leq C \iint_{\R ^2 \times \R ^2} 
		\bar{f} (-\Delta_x + 1) f \,dxdy. 
\end{align*}
For the first term we used that~\eqref{eq:smear coul exp} implies 
$$\left| \nablap \left( \eta  w_R  \right) \right|  \leq |\nabla \eta| |w_R| + |\eta| |\nabla w_R| \leq C$$ 
because $\eta \equiv 0$ in $B(0,R_0)$ and $\eta \equiv 1$ outside of $B (0,2 R_0)$. 
Gathering these estimates we obtain the desired operator bound.
\end{proof}

The three-body term (third line of~\eqref{eq:expand hamil}) is actually a 
pretty regular potential term, as shown in the following:

\begin{lemma}[\textbf{Three-body term}]\label{lem:three body}\mbox{}\\
	We have that, as operators on $L^2_\sym (\R^{6})$, 
	\begin{equation}\label{eq:three body}
		0 \leq \nablap w_R (x-y) \cdot \nablap w_R (x-z) \leq C (p_x ^2 + 1). 
	\end{equation}
\end{lemma}

The essential ingredient of the proof is the following three-particle Hardy 
inequality of~\cite[Lemma~3.6]{HofLapTid-08} 
(see also~\cite{Lundholm-15} for relevant methods and generalizations):

\begin{lemma}[\textbf{Three-body Hardy inequality}]\label{lem:Hardy}\mbox{}\\
	Let $d\geq 2$ and $u:\R ^{3d} \to \C$. 
	Let $\cR(x,y,z)$ be the circumradius of 
	the triangle with vertices $x,y,z \in \R^d$,
	and $\rho(x,y,z) := \sqrt{|x-y|^2 + |y-z|^2 + |z-x|^2}$. Then 
	$\cR^{-2} \le 9\rho^{-2}$ pointwise, and
	\begin{equation}\label{eq:Hardy}
	3(d-1)^2 \int_{\R ^{3d}} \frac{|u(x,y,z)| ^2}{\rho(x,y,z)^2} dxdydz \leq 
	\int_{\R ^{3d}} \left( \left| \nabla_x u \right| ^2 
		+ \left| \nabla_y u \right| ^2 
		+ \left| \nabla_z u \right| ^2\right) dx dy dz.  
	\end{equation}
\end{lemma}

\begin{proof}[Proof of Lemma~\ref{lem:three body}]
	Since we consider the operator as acting on symmetric wave functions it is 
	equivalent to estimate
	\begin{equation}\label{eq:cyclic}
		\sum_{\textup{cyclic in $x,y,z$}} \nablap w_R (x-y) \cdot \nablap w_R (x-z). 
	\end{equation}
	
	In general, let $x,y,z \in \R^d$ denote the vertices of a triangle and 
	$|x|_R := \max\{|x|,R\}$ a regularized distance. 
	Then we claim the following geometric fact:
	\begin{equation}\label{eq:geom three body}
		0 \le \sum_{\textup{cyclic in $x,y,z$}} \frac{x-y}{|x-y|_R^2} \cdot \frac{x-z}{|x-z|_R^2} 
		\le \frac{C}{\rho(x,y,z)^2}.
	\end{equation}
	Recalling~\eqref{eq:smear coul exp} this gives a control on the expression we are interested in. Equivalently, we shall prove that
	\begin{multline}\label{eq:geom three body mult}
		0 \le |y-z|_R^2 (x-y)\cdot(x-z) + |z-x|_R^2 (y-z)\cdot(y-x) + |x-y|_R^2 (z-x)\cdot(z-y) \\
		\le C\frac{|x-y|_R^2 |y-z|_R^2 |z-x|_R^2}{|x-y|^2 + |y-z|^2 + |z-x|^2},
	\end{multline}
	for some constant $C>0$ independent of $R$.
	
	Let us consider each of the different geometric configurations that may occur.
	In the case that all edge lengths of the triangle are greater than $R$,
	the cyclic expression that we wish to estimate in \eqref{eq:geom three body} 
	reduces to $\frac{1}{2}\cR(x,y,z)^{-2}$ 
	(see \cite[Lemma 3.2]{HofLapTid-08}), 
	which is clearly non-negative and bounded by
	$\frac{9}{2}\rho^{-2}$ from Lemma~\ref{lem:Hardy}.
	On the other hand, if all edge lengths are smaller than $R$ then
	the expression equals $\frac{1}{2R^4}\rho^2$
	(cf. \cite[Lemma 3.4]{HofLapTid-08}), for which we have 
	$0 \le \frac{1}{2R^4} \rho^2 \le \frac{9}{2} \rho^{-2}$
	since $\rho^2 \le 3R^2$.
	If two of the edges are short and one long, 
	say $|x-y|,|y-z| \le R$ and $|z-x| \ge R$,
	then the expression to be estimated in \eqref{eq:geom three body mult} 
	reads
	\begin{align*}
		&R^2(x-y)\cdot(x-z) + |z-x|^2 \, (y-z)\cdot(y-x) + R^2\underbrace{(z-x)\cdot(z-y)}_{-(z-y)\cdot(x-z)} \\
		&= R^2 \left( (x-y) - (z-y) ) \right) \cdot (x-z) 
			+ |x-z|^2 \, (y-z)\cdot(y-x) \\
		&= |x-z|^2 \Big( R^2 + (y-z)\cdot(y-x) \Big)
		\ \ge |x-z|^2 (R^2 - R^2)
		\ge 0.
	\end{align*}
	We furthermore have the upper bound
	$$
		|x-z|^2 \Big( R^2 + (y-z)\cdot(y-x) \Big)
		\le 2R^2 |x-z|^2,
	$$
	while the r.h.s. of \eqref{eq:geom three body mult} is larger than
	$$
		\frac{R^4|x-z|^2}{2R^2 + |x-z|^2} \ge \frac{1}{6} R^2|x-z|^2,
	$$
	using that $|x-z| \le |x-y| + |y-z| \le 2R$.
	
	This leaves the case that only one edge is short, say $|x-y| \le R$,
	and the others long, $|y-z|,|z-x| \ge R$.
	We thus consider the expression in \eqref{eq:geom three body mult}
	\begin{equation}\label{eq:geom mixed case}
		|y-z|^2 \, (x-y)\cdot(x-z) + |z-x|^2 \, (y-z)\cdot(y-x) + R^2 (z-x)\cdot(z-y).
	\end{equation}
	We will here use methods from \cite{Lundholm-15},
	namely the geometric (Clifford) algebra $\cG(\R^d)$ over $\R^d$
	(see \cite{LunSve-09} for a general introduction).
	In the case $d=2$ or $d=3$ one can think of this as the real algebra 
	generated by the Pauli matrices $\sigma_j$, 
	with scalar projection $\langle A \rangle_0 := \frac{1}{2}\tr A$
	and the embedding 
	of scalars (0-vectors)
	$1 \hookrightarrow \1$
	and of 1-vectors
	$\R^d \ni x \hookrightarrow \sum_{j=1}^d x_j\sigma_j \in \cG(\R^d)$,
	and with the product of two 1-vectors
	$xy = x \cdot y + x \wedge y$
	decomposing into a traceful symmetric scalar part 
	and a traceless antisymmetric bivector part.
	We have then,
	using tracelessness of the bivector parts of such products
	and the linearity and cyclicity of the trace,
	\begin{align*}
		&|y-z|^2 \, (x-y)\cdot(x-z) + |z-x|^2 \, (y-z)\cdot(y-x) \\
		&= \left\langle (y-z)^2(x-y)(x-z) + (z-x)^2(y-z)(y-x) \right\rangle_0 \\
		&= \left\langle (y-z)(x-y)(x-z)(y-z) + (y-x)(z-x)(z-x)(y-z) \right\rangle_0 \\
		&= \left\langle (y-z)(x-y)(z-x)(z-y) + (x-y)(z-x)(z-x)(z-y) \right\rangle_0 \\
		&= \left\langle \Big( (y-z)(x-y) + (z-x)(x-y) + 2(x-y)\wedge(z-x) \Big)(z-x)(z-y) \right\rangle_0 \\
		&= \left\langle (y-x)(x-y)(z-x)(z-y) \right\rangle_0 
		+ 2\big\langle (x-y)\wedge(z-x) \, (z-x)(z-y) \big\rangle_0 \\
		&= -\left\langle (x-y)^2(z-x)(z-y) \right\rangle_0 
		+ 2\big\langle \!\!\!\underbrace{(x-y)}_{z-y-(z-x)}\!\!\!\wedge(z-x) \, (z-x)\wedge(z-y) \big\rangle_0 \\
		&= -|x-y|^2 \, (z-x)\cdot(z-y) + 2\langle B^\dagger B \rangle_0,
	\end{align*}
	with $B := (z-x) \wedge (z-y)$ and its Hermite conjugate
	$B^\dagger = (z-y) \wedge (z-x)$.
	In the fourth and fifth steps we used $xy = yx + 2x \wedge y$ 
	for the second term and then $(y-z) + (z-x) = y-x$,
	while for the final steps we again used the properties of the trace
	and that $B^\dagger B = |B|^2$ is scalar.
	Thus, the expression \eqref{eq:geom mixed case} we wish to estimate equals
	$$
		\left( R^2 - |x-y|^2 \right) (z-x)\cdot(z-y) + 2|B|^2 \ \ge 0,
	$$
	where for the lower bound we also used that $(z-x) \cdot (z-y) \ge 0$ 
	since $x-y$ is the shortest edge.
	For an upper bound we can use permutation invariance (cf. \cite[Proposition 15]{Lundholm-15}) of
	$$|B| = |(x-y) \wedge (x-z)| \le |x-y||x-z|,$$ 
	and for example that $|y-z| \le R + |x-z| \le 2|x-z|$. Hence
	$$
		\left( R^2 - |x-y|^2 \right) (z-x)\cdot(z-y) + 2|B|^2
		\le 4R^2|x-z|^2,
	$$
	while for the r.h.s. of \eqref{eq:geom three body mult},
	with analogously $|x-z| \le 2|y-z|$,
	$$
		\frac{R^2|y-z|^2|z-x|^2}{R^2 + |y-z|^2 + |z-x|^2} 
		\ge \frac{R^2|y-z|^2|z-x|^2}{6|y-z|^2}
		= \frac{1}{6} R^2|x-z|^2.
	$$
	We also remark that the non-negativity of \eqref{eq:geom three body}
	is in general false if $|\cdot|_R$ 
	is replaced by an arbitrary radial function, 
	as can be checked when taking e.g. $|x|_R = e^{|x|^2/2}$.

	Finally, the estimate \eqref{eq:three body} follows simply by applying
	Lemma~\ref{lem:Hardy} with $d=2$ to \eqref{eq:geom three body}
	and using the symmetry of functions in $L^2_\sym(\R^6)$.
\end{proof}

\subsection{A priori bound for the ground state}\label{sec:ap bound}

For the estimates of the previous subsection to apply efficiently, we need an a priori bound on ground states (or approximate ground states) of 
the $N$-body Hamiltonian~\eqref{eq:hamil}, provided in the following:

\begin{proposition}[\textbf{A priori bound for many-body ground states}]\label{pro:a priori}\mbox{}\\
Let $\Psi_N \in L^2_\sym (\R ^{2N})$ be a (sequence of) approximate ground 
states for $H_N ^R$, that is,
$$
	\langle \Psi_N, H_N ^R \Psi_N \rangle \leq E^R (N) (1+o(1)) \mbox{ when } N\to \infty.
$$
Denote by $\gamma_N ^{(1)}$ the associated sequence of one-body density matrices. 
In the regime~\eqref{eq:scale alpha}, assuming a bound $R\geq N ^{-\eta}$ for 
some $\eta > 0$ independent of $N$, we have
\begin{equation}\label{eq:ap bound}
\tr\left[ \left( p^2 + V \right) \gamma_N ^{(1)}\right] \leq C(1+\beta^2),
\end{equation}
where $C$ is a constant independent of $\beta$, $N$ and $R$.
\end{proposition}

\begin{proof}
We proceed in two steps.

\noindent\textbf{Step 1.} Using a trial state $u ^{\otimes N}$ 
with $u = |u|\in C^\infty_c (\R ^2)$, we easily obtain from \eqref{eq:ener dens mat}
and the above bounds (note that the $R$-divergent mixed two-body term is zero on such a $u$, and that the singular two-body term gives a lower-order contribution)
\begin{equation}\label{eq:ap up bound}
 E^R (N) \leq C (1+\beta^2) N.  
\end{equation}
Next we use the diamagnetic inequality~\cite[Theorem 7.21]{LieLos-01} in each 
variable to obtain
\begin{align*}
\langle \Psi_N, H_N ^R \Psi_N \rangle 
&= \sum_{j=1} ^N \int_{\R ^{2N}} \left( \left| \left( -i \nabla_j + \alpha \bA^R_j \right) \Psi_N \right| ^2 + V(x_j) |\Psi_N| ^2\right) dx_1\ldots dx_N\\
&\geq \sum_{j=1} ^N \int_{\R ^{2N}} \left( \left| \nabla_j | \Psi_N | \right| ^2 + V(x_j) |\Psi_N| ^2\right) dx_1\ldots dx_N.
\end{align*}
We deduce the bound 
\begin{equation}\label{eq:pre ap bound}
\tr \left[ \left( p ^2 + V \right) \gamma_{N,+} ^{(1)}\right] \leq C (1+\beta^2),
\end{equation}
where we denote 
$$\gamma_{N,+} ^{(k)} := \tr_{k+1 \to N} \left[ \left| \, |\Psi_N| \, \right\rangle \left\langle \, |\Psi_N| \, \right| \right] $$ 
the reduced $k$-body density matrix of $|\Psi_N|$.

\medskip

\noindent\textbf{Step 2.} Next we expand the Hamiltonian and use the 
Cauchy-Schwarz inequality for operators to obtain
\begin{align*}
 H_N ^R &= \sum_{j=1}^N \left( p_j ^2 + \alpha p_j \cdot \bA_j^R + \alpha \bA_j^R \cdot p_j + \alpha^2 |\bA_j ^R| ^2 + V(x_j) \right)\\
 &\geq \sum_{j=1}^N \left( (1-2\delta ^{-1}) p_j ^2 + (1-2\delta) \alpha^2 |\bA_j ^R| ^2 + V(x_j) \right)\\
 &= \sum_{j=1}^N \left( \frac{1}{2} ( p_j ^2 + V(x_j)) -7 \frac{\beta^2}{(N-1)^2} |\bA_j ^R|^2 \right),
\end{align*}
choosing $\delta = 4$. Thus, using~\eqref{eq:ap up bound} we have 
\begin{equation}\label{eq:interm}
	\tr\left[ \left( p^2 + V \right) \gamma_N ^{(1)}\right] 
	\leq C(1+\beta^2) + \frac{C\beta^2}{N(N-1)^2} \left\langle \Psi_N , \sum_{j=1} ^N |\bA_j ^R| ^2  \Psi_N \right\rangle. 
\end{equation}
Then, since the last term in the right-hand side is purely a potential term
$$ \left\langle \Psi_N , \sum_{j=1} ^N |\bA_j ^R| ^2  \Psi_N \right\rangle = \left\langle |\Psi_N| , \sum_{j=1} ^N |\bA_j ^R| ^2  |\Psi_N| \right\rangle.$$
We then expand the squares as in~\eqref{eq:ener dens mat}, and use 
Lemmas~\ref{lem:sing two body} and~\ref{lem:three body} to obtain for any 
$\eps >0$
\begin{align*}
&\frac{1}{N(N-1)^2} \left\langle |\Psi_N| , \sum_{j=1} ^N |\bA_j ^R| ^2  |\Psi_N| \right\rangle 
\leq C \tr\left[ \nablap w_R (x_1-x_2) \cdot \nablap w_R (x_1-x_3) \gamma_{N,+} ^{(3)}\right] \\
&\qquad + C N^{-1} \tr\left[ |\nabla w_R (x_1-x_2)| ^2 \gamma_{N,+} ^{(2)}\right]\\
&\quad \leq C \tr\left[ (p_1^2 +1) \otimes \1 \otimes \1 \gamma_{N,+} ^{(3)}\right] + C_\eps R ^{-\eps} N ^{-1} \tr\left[ (p_1 ^2 + 1) \otimes \1 \gamma_{N,+} ^{(2)}\right] \\
&\quad \leq C \left( 1 + C_{\eps} N ^{-1} R ^{-\eps} \right) \tr\left[ (p_1 ^2 +1) \gamma_{N,+} ^{(1)}\right] .
\end{align*}
Inserting the estimate~\eqref{eq:pre ap bound} and recalling that we assume 
$R\geq N ^{-\eta}$ we conclude the proof by going back to~\eqref{eq:interm}.
\end{proof}

\section{Mean-field limit}\label{sec:MF lim}

We now turn to the study of the mean-field limit per se. 
The strategy is the same as in~\cite{LewNamRou-14}, but the peculiarities of 
the anyon Hamiltonian add some important twists, and we shall rely heavily on the 
estimates of the preceding section.

\subsection{Preliminaries}\label{sec:MF pre}

We first recall some constructions from~\cite{Lewin-11,LewNamRou-14}.

\medskip

\noindent\textbf{Energy cut-off.} We denote by $P$ the spectral projector of $-\Delta + V $ below a given 
(large) energy cut-off $\Lambda$ that we shall optimize over in the end: 
\begin{equation}\label{eq:projector}
P := \1_{h \leq \Lambda }, \quad h = -\Delta + V. 
\end{equation}
Let
$$ N_\Lambda = \dim (P L^2 (\R ^2)) $$
be the number of energy levels obtained this way, and recall the 
following Cwikel-Lieb-Rozenblum type inequality, proved by well-known methods, 
as in~\cite[Lemma~3.3]{LewNamRou-14}:

\begin{lemma}[\textbf{Number of energy levels below the cut-off}] \label{lem:CLR}\mbox{}\\
For $\Lambda$ large enough we have
\begin{equation}\label{eq:CLR}
N_\Lambda \leq C \Lambda^{1 + 2/s}. 
\end{equation}
\end{lemma}

We shall also denote 
$$ Q = \1 - P$$
the orthogonal projector onto excited energy levels.

\medskip

\noindent\textbf{Localization in Fock space.} We quickly recall the procedure of geometric localization, following the 
notation of~\cite{Lewin-11}. Let $\gamma_{N}$  be an arbitrary $N$-body (mixed) state. 
Associated with the given projector $P$, there is a localized state $G_{N}^P$ 
in the Fock space 
$$\cF(\gH)=\C\oplus \gH \oplus\gH^2\oplus\cdots$$
of the form
\begin{equation}
G_N ^{P} = G_{N,0}^ {P} \oplus G_{N,1}^{P} \oplus\cdots\oplus G_{N,N}^{P} \oplus0\oplus\cdots 
\label{eq:def_localization}
\end{equation}
with the property that its reduced density matrices satisfy
\begin{equation}
P ^{\otimes n} \gamma^{(n)}_{N} P ^{\otimes n} = \left(G_N ^{P}\right)  ^{(n)}={N\choose n}^{-1}\sum_{k=n}^N{k\choose n}\tr_{n+1\to k}\left[G^{P}_{N,k}\right]
\label{eq:localized-DM} 
\end{equation}
for any $0 \leq n \leq N$. Here we use the convention that 
$$\gamma_N^{(n)}:=\Tr_{n+1\to N} [\gamma_N],$$
which differs from the convention of~\cite{Lewin-11}, whence the different 
numerical factors in~\eqref{eq:localized-DM}. 
We also have a localized state $G_N ^Q$ corresponding to the projector $Q$, 
which is defined similarly.

The relations \eqref{eq:localized-DM} determine the localized states 
$G_{N}^P, G_{N}^Q$ uniquely and they ensure that $G_N ^P$ and $G_N ^Q$ are 
(mixed) states on the projected Fock spaces $\cF (P \gH)$ and $\cF (Q \gH)$, 
respectively:
\begin{equation}
\label{eq:nomalization-localized-state}
\sum_{k=0}^N \tr \left[ G_{N,k}^{P/Q} \right] =1. 
\end{equation}

\medskip

\noindent\textbf{de Finetti measure for the projected state.} 
We will apply the quantitative de Finetti Theorem in finite dimensional spaces 
of~\cite{ChrKonMitRen-07,Chiribella-11,Harrow-13,LewNamRou-13b} 
to the localized state $G_N^P$, in order to approximate its three-body 
density matrix. 
The following is the equivalent of~\cite[Lemma~3.4]{LewNamRou-14} and the 
proof is exactly similar:

\begin{lemma} [\textbf{Quantitative quantum de Finetti for the localized state}] \label{lem:deF-localized-state}\mbox{}\\
Let $\gamma_{N}$ be an arbitrary $N$-body (mixed) state. Define
\bq \label{eq:def-mu-N-localized}
d\mu_N(u) := \sum_{k=3}^N {N \choose 3} ^{-1}  {k \choose 3} d\mu_{N,k}(u), \quad d\mu_{N,k}(u) :=  \dim (P \gH)_\sym^k \pscal{u^{\otimes k},G_{N,k}^P u^{\otimes k}} du
\eq
and 
\begin{equation}\label{eq:deF state}
\widetilde{\gamma}_N ^{(3)} := \int_{SP\gH} |u^{\otimes 3}\rangle \langle u^{\otimes 3}| d\mu_N(u).
\end{equation}
Then there is a constant $C>0$ such that for every $N\in\NN$ and $\Lambda>0$, 
we have
\begin{equation}\label{eq:deF estim}
\Tr \left| P^{\otimes 3} \gamma_{N}^{(3)} P^{\otimes 3} - \widetilde{\gamma}_N ^{(3)} \right| \le \frac{C N_\Lambda}{N}. 
\end{equation}
\end{lemma}

\subsection{Truncated Hamiltonian}\label{sec:MF trunc}

For an energy lower bound we are first going to roughly bound some terms in 
the Hamiltonian. Let us introduce the effective three-body Hamiltonian
\begin{multline}\label{eq:3 body hamil}
\tH_3 ^R := \frac{1}{3} \left( h_1 + h_2 + h_3 \right) + \frac{\beta}{6} \sum_{1 \leq j\neq k \leq 3} \left( p_j \cdot \nablap w_R (x_j-x_k) + \nablap w_R (x_j-x_k) \cdot p_j \right)
\\+ \beta ^2  \nablap w_R (x_1-x_2) \cdot \nablap w_R (x_1 -x_3) 
\end{multline}
where $h_i$ is understood to act on the $i$-th variable (recall that $h= -\Delta +V$). For shortness we denote
$$ W_2 = p_1 \cdot \nablap w_R (x_1-x_2) + \nablap w_R (x_1- x_2) \cdot p_1$$
the two-body part of $\tH_3 ^R$, and 
$$
	W_3 = \nablap w_R (x_1-x_2) \cdot \nablap w_R (x_1 -x_3)
$$
its three-body part.
With this notation 
$$
	\tH_3 ^R := \frac{1}{3} \left( h_1 + h_2 + h_3 \right) 
	+ \frac{\beta}{6} \sum_{1 \leq i\neq j \leq 3} W_2 (i,j) + \beta ^2 W_3 
$$
where $W_2 (i,j)$ acts on variables $i$ and $j$. 
Also note that for $\|u\|=1$, 
by \eqref{eq:two body term}, \eqref{eq:three body term},
$$
	\langle u^{\otimes 3}, \tH_3^R u^{\otimes 3} \rangle 
	= \cEAF_R[u] \ \ge \ \EAF_R.
$$
We bound the full energy from below in terms of a projected version 
of~$\tH_3 ^R$:

\begin{proposition}[\textbf{Truncated three-body Hamiltonian}]\label{pro:eff Hamil}\mbox{}\\
Let $\Psi_N$ be a (sequence of) approximate ground state(s) for $H_N ^R$ with 
associated reduced density matrices $\gamma_N ^{(k)}$.  
Then, for any $\eps>0$ and $R$ small enough,
\begin{multline}\label{eq:ener 3 body}
	\frac{1}{N}\langle \Psi_N, H_N ^R \Psi_N \rangle \geq \tr\left[ \tH_3 ^R P ^{\otimes 3} \gamma_N ^{(3)} P ^{\otimes 3}\right] 
	+ C_\beta \Lambda \tr[Q \gamma_N ^{(1)}] 
	\\-C_\beta \left( \frac{1}{N} + \frac{C_\eps}{\sqrt{\Lambda} R ^{1+\eps}} + \frac{1}{\Lambda R ^2}\right). 
\end{multline}
\end{proposition}

\begin{proof}
We proceed in several steps. 

\noindent\textbf{Step 1.} We first claim that
\begin{equation}\label{eq:ener 3 body pre}
\frac{1}{N}\langle \Psi_N, H_N ^R \Psi_N \rangle \geq \tr\left[ \tH_3 ^R \gamma_N ^{(3)} \right] -C_\beta N ^{-1}. 
\end{equation}
To see this, we start from~\eqref{eq:ener dens mat}. 
For a lower bound we drop the term on the fourth line, which is positive. 
Then one only has to correct the $N$-dependent factors in front of the third line. 
The term we have to drop to obtain~\eqref{eq:ener 3 body pre} is bounded as 
$$
	\beta ^2 \left| 1 - \frac{N-2}{N-1}\right| \left|\tr\left[ \left(\nablap w_R (x_1-x_2) \cdot \nablap w_R (x_1-x_3) \right) \gamma_N ^{(3)} \right] \right| 
	\leq C_\beta N ^{-1}
$$
upon using the a priori bound~\eqref{eq:ap bound} combined with~\eqref{eq:three body}. 

\medskip

\noindent\textbf{Step 2.} We next proceed to bound the right-hand side 
of~\eqref{eq:ener 3 body pre} from below in terms of a localized version of 
$\tH_3 ^R$ and remainder terms to be estimated in the next step.
 We shall need the projectors
\begin{align*}
 \Pi_2 &= \1 ^{\otimes 2} - P ^{\otimes 2}\\ 
\Pi_3 &= \1 ^{\otimes 3} - P ^{\otimes 3}
\end{align*}
and make a repeated use of the inequality
\begin{equation}\label{eq:op CS}
A B C + C B A \geq - \eps A |B| A - \eps ^{-1} C |B| C, \quad \eps > 0,
\end{equation}
for any self-adjoint operators $A,B,C$.

We claim that
\begin{align}\label{eq:ener loc 3 body}
\tr\left[ \tH_3 ^R \gamma_N ^{(3)} \right] \geq & \tr\left[  \tH_3 ^R P ^{\otimes 3}\gamma_N ^{(3)} P ^{\otimes 3}\right] + \tr\left[ h Q \gamma_N ^{(1)} Q \right] \nonumber
\\&-  |\beta| (3+\delta_1) \tr\left[ \Pi_2 |W_2| \Pi_2 \gamma_N ^{(2)}\right] - |\beta| \delta_1 ^{-1} \tr\left[ P ^{\otimes 2}  |W_2| P ^{\otimes 2}\gamma_N ^{(2)} \right]\nonumber
\\&-2 |\beta| \tr\left[ P ^{\otimes 2 } \otimes Q |W_2(1,2)| P ^{\otimes 2} \otimes Q \gamma_N ^{(3)}\right] \nonumber
\\&- \beta ^2 (1+\delta_2) \tr\left[ \Pi_3 |W_3| \Pi_3 \gamma_N ^{(3)}\right] -  \beta ^2 \delta_2 ^{-1} \tr\left[ P ^{\otimes 3}  |W_3| P ^{\otimes 3} \gamma_N ^{(3)}\right]
\end{align}
where $\delta_1$ and $\delta_2$ 
are two positive parameters to be chosen later on. 

To prove~\eqref{eq:ener loc 3 body}, first note that
$$\tr\left[ \tH_3 ^R \gamma_N ^{(3)} \right] = \tr\left[ h \gamma_N ^{(1)} \right] + \frac{\beta}{2} \tr\left[ W_2 \gamma_N ^{(2)}\right] + \beta ^2 \tr\left[ W_3 \gamma_N ^{(3)}\right].$$
Then, for the one-body term we have
\begin{align*}
 \tr\left[ h \gamma_N ^{(1)} \right] &= \tr\left[ P h P \gamma_N ^{(1)} \right] + \tr\left[ Q h Q \gamma_N ^{(1)} \right]\\
 &\geq \frac{1}{3} \tr\left[ P ^{\otimes 3} \left( h_1 + h_2 + h_3 \right) P^{\otimes 3} \gamma_N ^{(3)} \right] + \tr\left[ Q h Q \gamma_N ^{(1)} \right]
\end{align*}
using that $h$ commutes with $P$ and $Q$, $PQ = QP = 0$ and the fact that 
$h$ is a positive operator.

For the two-body term we write
\begin{align*} 
\tr\left[ W_2 \gamma_N ^{(2)}\right] &= \tr\left[ P ^{\otimes 3} W_2 (1,2) P ^{\otimes 3} \gamma_N ^{(3)}\right] + \tr\left[ \Pi_3 W_2 (1,2) \Pi_3 \gamma_N ^{(3)}\right] 
\\&+ \tr\left[ \left(P ^{\otimes 3} W_2 (1,2) \Pi_3 + \Pi_3 W_2 (1,2) P ^{\otimes 3}\right) \gamma_N ^{(3)}\right]. 
\end{align*} 
Next, since 
\begin{equation}\label{eq:expand Pi3}
 \Pi_3 = P ^{\otimes 2} \otimes Q + \Pi_2 \otimes P + \Pi_2 \otimes Q, 
\end{equation}
and $W_2 (1,2)$ only acts on the first two variables, this simplifies into  
\begin{align*} 
\tr\left[ W_2 \gamma_N ^{(2)}\right] &= \tr\left[ P ^{\otimes 3} W_2 (1,2) P ^{\otimes 3} \gamma_N ^{(3)}\right] + \tr\left[ \Pi_3 W_2 (1,2) \Pi_3 \gamma_N ^{(3)}\right] 
\\&+ \tr\left[ \left(P ^{\otimes 3} W_2 (1,2) \Pi_2 \otimes P + \Pi_2 \otimes P W_2 (1,2) P ^{\otimes 3}\right) \gamma_N ^{(3)}\right]
\\&\geq \tr\left[ P ^{\otimes 3} W_2 (1,2) P ^{\otimes 3} \gamma_N ^{(3)}\right] + \tr\left[ \Pi_3 W_2 (1,2) \Pi_3 \gamma_N ^{(3)}\right] 
\\&-\delta_1  \tr\Big[\Pi_2 \otimes P  |W_2 (1,2)| \Pi_2 \otimes P \gamma_N ^{(3)}\Big] - \delta_1 ^{-1} \tr\left[ P ^{\otimes 3} |W_2 (1,2)| P ^{\otimes 3} \gamma_N ^{(3)}\right]
\end{align*} 
where we use~\eqref{eq:op CS} to obtain the lower bound. Then, using~\eqref{eq:expand Pi3} and~\eqref{eq:op CS} again for the second term of the right-hand side, as well as $P,Q\leq \1$, we get 
\begin{align*} 
\tr\left[ W_2 \gamma_N ^{(2)}\right] &\geq \tr\left[ P ^{\otimes 3} W_2 (1,2) P ^{\otimes 3} \gamma_N ^{(3)}\right] - \delta_1 ^{-1} \tr\left[ P ^{\otimes 2} |W_2| P ^{\otimes 2} \gamma_N ^{(2)}\right]
\\&- \left( 3 + \delta_1 \right) \tr\left[ \Pi_2 |W_2| \Pi_2 \gamma_N ^{(2)}\right] - 2 \tr\left[ P ^{\otimes 2 } \otimes Q |W_2(1,2)| P ^{\otimes 2} \otimes Q \gamma_N ^{(3)}\right].
\end{align*}
Finally, the three-body term is dealt with similarly:
\begin{align*}
\tr\left[ W_3 \gamma_N ^{(3)}\right]  &= \tr\left[ \left( P ^{\otimes 3} + \Pi_3 \right) W_3 \left( P ^{\otimes 3} + \Pi_3 \right) \gamma_N ^{(3)}\right]\\
&\geq \tr\left[ P ^{\otimes 3} W_3 P ^{\otimes 3} \gamma_N ^{(3)}\right] - (1+\delta_2) \tr\left[ \Pi_3 |W_3| \Pi_3 \gamma_N ^{(3)}\right] -  \delta_2 ^{-1} \tr\left[ P ^{\otimes 3}  |W_3| P ^{\otimes 3} \gamma_N ^{(3)}\right]
\end{align*}
using~\eqref{eq:op CS} again. All in all, using also the symmetry of 
$\gamma_N ^{(3)}$, we obtain~\eqref{eq:ener loc 3 body}.

\medskip

\noindent\textbf{Step 3.} We next estimate the remainder terms 
in~\eqref{eq:ener loc 3 body}. First we note that 
\begin{align}\label{eq:control loc 1}
\tr\left[ h Q \gamma_N ^{(1)} Q \right] &\geq \frac{\Lambda}{2} \tr \left[ Q \gamma_N ^{(1)} Q\right] + \frac{\sqrt{\Lambda}}{2} \tr\left[ \sqrt{h} Q \gamma_N ^{(1)} Q\right] \nonumber\\
&\geq \frac{\Lambda}{4} \tr \left[ Q \gamma_N ^{(1)} Q\right] + \frac{\Lambda}{20} \tr\left[ \Pi_3 \gamma_N ^{(3)} \Pi_3 \right] + \frac{\sqrt{\Lambda}}{4} \tr\left[ \sqrt{h}_1 \Pi_2 \gamma_N ^{(2)} \Pi_2 \right].
\end{align}
The first inequality is just the definition of $Q$, and to see the second one we first write  
\begin{align*}
2 \tr\left[ \sqrt{h} Q \gamma_N ^{(1)} Q\right] 
=&  \tr\left[ \sqrt{h}_1 Q \otimes \1 \gamma_N ^{(2)} \right] + \tr\left[  \sqrt{h}_2 \1 \otimes Q \gamma_N ^{(2)} \right]\\
=& \tr\left[ \sqrt{h}_1 \left(Q \otimes P + Q \otimes Q \right) \gamma_N ^{(2)} \right] + \tr\left[ \sqrt{h}_2 \left(P \otimes Q + Q \otimes Q \right) \gamma_N ^{(2)} \right]\\
=& \tr\left[ \sqrt{h}_1 \left(Q \otimes P + Q \otimes Q \right) \gamma_N ^{(2)} (P ^{\otimes 2} + \Pi_2)\right] 
\\&+ \tr\left[ \sqrt{h}_2 \left(P \otimes Q + Q \otimes Q \right) \gamma_N ^{(2)} (P ^{\otimes 2} + \Pi_2)\right]\\
=& \tr\left[ \sqrt{h}_1 \Pi_2 \gamma_N ^{(2)} \Pi_2\right] + \tr\left[ \sqrt{h}_2 Q \otimes Q \gamma_N ^{(2)} Q\otimes Q \right] 
\\&+ \tr\left[ \left(\sqrt{h}_2 - \sqrt{h}_1\right) P \otimes Q \gamma_N ^{(2)} P\otimes Q \right]\\
\geq& \tr\left[ \sqrt{h}_1 \Pi_2 \gamma_N ^{(2)} \Pi_2\right],
\end{align*}
where we use repeatedly the cyclicity of the trace and the fact that 
$\sqrt{h}$ commutes with $P$ and $Q$, along with the fact that 
$PQ = QP = 0$ and 
$$ \Pi_2 = \1 ^{\otimes 2} - P ^{\otimes 2} = Q\otimes Q + P\otimes Q + Q \otimes P.$$
In the last step we also use that as operators
$$ \sqrt{h}_2 P\otimes Q \geq \sqrt{\Lambda} P\otimes Q \geq \sqrt{h}_1 P\otimes Q$$
by definition of the projectors $P$ and $Q$. 
This gives the third term in the right-hand side of~\eqref{eq:control loc 1}. 
The second one arises from similar considerations:
\begin{align*}
\tr \left[ \Pi_3 \gamma_N ^{(3)} \right] &=  \tr \left[ \Pi_2 \otimes Q \gamma_N ^{(3)} \right] + \tr \left[ P ^{\otimes 2} \otimes Q \gamma_N ^{(3)} \right] + \tr \left[ \Pi_2 \otimes P \gamma_N ^{(3)}\right] \\
&\leq 2\tr \left[ Q \gamma_N ^{(1)} \right] + \tr \left[ \Pi_2 \gamma_N ^{(2)}\right]
\\&= 2\tr \left[ Q \gamma_N ^{(1)} \right] + \tr \left[ (P\otimes Q + Q \otimes P + Q \otimes Q ) \gamma_N ^{(2)}\right]
\\&\leq 5 \tr \left[ Q \gamma_N ^{(1)} \right].
\end{align*}

Next, using~\eqref{eq:mix two body bis}, we have
\begin{equation}\label{eq:high two body}
	\tr\left[ \Pi_2 |W_2| \Pi_2 \gamma_N ^{(2)}\right] \leq \frac{C}{R} \tr\left[ |p_1| \Pi_2 \gamma_N ^{(2)} \Pi_2 \right] \leq \frac{C}{R} \tr\left[ \sqrt{h}_1 \Pi_2 \gamma_N ^{(2)} \Pi_2 \right]
\end{equation}
by operator monotonicity of the square-root, and by~\eqref{eq:sup w_R} 
\begin{equation}\label{eq:high three body}
	\tr\left[ \Pi_3 |W_3| \Pi_3 \gamma_N ^{(3)}\right] \leq \frac{C}{R ^2} \tr\left[  \Pi_3 \gamma_N ^{(3)} \Pi_3 \right].
\end{equation}
Moreover, using~\eqref{eq:mix two body bis} again we get
$$
\tr\left[ P ^{\otimes 2 } \otimes Q |W_2(1,2)| P ^{\otimes 2} \otimes Q \gamma_N ^{(3)}\right] \leq \frac{C\sqrt{\Lambda}}{R} \tr\left[ P ^{\otimes 2} \otimes Q \gamma_N ^{(3)}\right] \leq \frac{C\sqrt{\Lambda}}{R} \tr\left[ Q \gamma_N ^{(1)} \right]
$$
so that, combining with~\eqref{eq:control loc 1}, choosing for some small 
fixed $c_1,c_2>0$
$$
\delta_1 = c_1 \sqrt{\Lambda} R, \quad \delta_2 = c_2 \Lambda R^2
$$ 
and $\Lambda$ large enough (i.e. $\Lambda R^{2} > c$ for $c$ large enough), 
we get 
\begin{multline*}
\tr\left[ h Q \gamma_N ^{(1)} Q \right] -  |\beta| (3+\delta_1) \tr\left[ \Pi_2 |W_2| \Pi_2 \gamma_N ^{(2)}\right] -2 |\beta| \tr\left[ P ^{\otimes 2 } \otimes Q |W_2(1,2)| P ^{\otimes 2} \otimes Q \gamma_N ^{(3)}\right]
\\-\beta ^2 (1+\delta_2) \tr\left[ \Pi_3 |W_3| \Pi_3 \gamma_N ^{(3)}\right] \geq C \tr\left[ h Q \gamma_N ^{(1)} Q \right]
\end{multline*}
for some fixed constant $C>0$. 
Then, inserting in~\eqref{eq:ener loc 3 body}, we deduce 
\begin{align}\label{eq:ener loc 3 body 2}
	\tr\left[ \tH_3 ^R \gamma_N ^{(3)} \right] &\geq \tr\left[  \tH_3 ^R P ^{\otimes 3}\gamma_N ^{(3)} P ^{\otimes 3}\right] + C \tr\left[ h Q \gamma_N ^{(1)} Q \right] \nonumber
	\\& - \frac{C}{\sqrt{\Lambda } R} \tr\left[ P ^{\otimes 2}  |W_2| P ^{\otimes 2}\gamma_N ^{(2)} \right]\nonumber
	\\&- \frac{C}{\Lambda  R ^2} \tr\left[ P ^{\otimes 3}  |W_3| P ^{\otimes 3} \gamma_N ^{(3)}\right].
\end{align}
But, using~\eqref{eq:mix two body eps}, \eqref{eq:ap bound}, and 
$\tr \gamma_N^{(k)} = 1$,
\begin{align*}
	\tr\left[ P ^{\otimes 2} |W_2| P ^{\otimes 2}\gamma_N ^{(2)} \right] 
	&\leq C_\eps R ^{-\eps} \tr\left[ P ^{\otimes 2} (p_1^2+1) P^{\otimes 2}\gamma_N ^{(2)} \right] 
	\\&\leq C_\eps R ^{-\eps} \tr\left[ (p_1^2+1) \gamma_N ^{(1)} \right] 
	\leq C_\eps R ^{-\eps},
\end{align*}
whereas, using~\eqref{eq:three body} and~\eqref{eq:ap bound} again
\begin{align*}
	\tr\left[ P ^{\otimes 3} |W_3| P^{\otimes 3} \gamma_N ^{(3)}\right] 
	&\leq C \tr\left[ P ^{\otimes 3} (p_1^2+1) P^{\otimes 3}\gamma_N ^{(3)} \right] 
	\\&\leq C \tr\left[ (p_1^2+1) \gamma_N ^{(1)} \right] \leq C,
\end{align*}
which completes the proof.
\end{proof}

\subsection{Energy bounds}\label{sec:low bound}

In this subsection we prove the 
energy bounds establishing \eqref{eq:main ener}.
The upper bound is 
obtained as usual by testing against a factorized trial state.
Namely, taking $\Psi_N = (\uAF_R)^{\otimes N}$ in \eqref{eq:ener dens mat}
with $\uAF_R$ a normalized minimizer of $\cEAF_R$,
and using \eqref{eq:two body term}, \eqref{eq:three body term},
Lemmas~\ref{lem:sing two body}, 
\ref{lem:mix two body}, \ref{lem:three body},
and the diamagnetic inequality \eqref{eq:af_diamagnetic},
one finds
\begin{equation} \label{eq:ener upper bound}
	\frac{E^R(N)}{N} \le \cEAF_R[\uAF_R] +
		(\cEAF_R[\uAF_R] + 1) \left(
		\frac{C \beta^2}{N} + \frac{C_\eps \beta^2 R^{-\eps}}{N}
		\right)
	= \EAF_R + o(1) \to \EAF,
\end{equation}
as $R \sim N^{-\eta}$ with $N \to \infty$,
where we also used Proposition~\ref{prop:af_limit}.
Note that for this upper bound we can allow any rate $0<\eta<\infty$,
and may even take $\beta = \beta(N) \to \infty$.

\medskip

For the lower bound,
inserting~\eqref{eq:deF estim} in the estimate of 
Proposition~\ref{pro:eff Hamil} we get 
for any sequence of ground states $\Psi_N$ that
\begin{align}\label{eq:ener pre final}
\frac{1}{N}\langle \Psi_N, H_N ^R \Psi_N \rangle 
\geq& \tr\left[ \tH_3 ^R \widetilde{\gamma}_N ^{(3)} \right] + C \Lambda \tr[Q \gamma_N ^{(1)}] - C \frac{N_\Lambda}{N} \norm{P ^{\otimes 3} \tH_3 ^R P ^{\otimes 3}}\nonumber
\\&-C \left( \frac{1}{N} + \frac{C_\eps}{\sqrt{\Lambda} R ^{1+\eps}} + \frac{1}{\Lambda R ^2}\right)\nonumber
\\ \geq& \tr\left[ \tH_3 ^R \widetilde{\gamma}_N ^{(3)} \right] + C \Lambda \tr[Q \gamma_N ^{(1)}]- C \frac{\Lambda ^{2+2/s}}{N} \left( 1 + |\log R| \right)\nonumber
\\&-C \left( \frac{1}{N} + \frac{C_\eps}{\sqrt{\Lambda} R ^{1+\eps}} + \frac{1}{\Lambda R ^2}\right).
\end{align}
We have here used Estimates~\eqref{eq:mix two body} and~\eqref{eq:three body} 
from Subsection~\ref{sec:op bounds}, along with 
$$P p^2 P \leq P h P \leq \Lambda$$
to bound the operator norm of 
$P ^{\otimes 3} \tH_3 ^R P ^{\otimes 3}$ 
and Lemma~\ref{lem:CLR} to bound $N_\Lambda$. 

\medskip

\noindent\textbf{Main term.} Since by definition $\widetilde{\gamma}_N^{(3)}$ 
is a superposition of tensorized states we get  
$$ 
	\tr\left[ \tH_3 ^R \widetilde{\gamma}_N ^{(3)} \right] 
	\geq \EAF_R \tr\left[ \widetilde{\gamma}_N ^{(3)} \right].
$$
We then denote 
\begin{equation}\label{eq:def lambda}
 \lambda := \tr \left[ P\gamma_N^{(1)} \right] = \sum_{k=0}^N \frac{k}{N} \tr\left[ G_{N,k} ^P\right] 
\end{equation}
the fraction of $P$-localized particles. Using the simple estimate 
$$\left|\frac{k}{N}\frac{k-1}{N-1}\frac{k-2}{N-2} - \frac{k^3}{N^3}\right| \leq C N ^{-1} \mbox{ for } 0\leq k \leq N,$$
it follows from~\eqref{eq:nomalization-localized-state},~\eqref{eq:def-mu-N-localized},~\eqref{eq:deF state},
and Jensen's inequality that  
\begin{align}\label{eq:lambda-lambda3}
\tr\left[ \widetilde{\gamma}_N^{(3)} \right] 
&= \int_{SP\gH} d\mu_N  
 = \sum_{k=3} ^N {N \choose 3} ^{-1}  {k \choose 3} \tr\left[ G_{N,k} ^P\right]\nonumber
\\&\geq \sum_{k=0} ^N \left( \frac{k}{N}\right) ^3 \tr\left[ G_{N,k} ^P\right] - O(N^{-1})\nonumber
\\&\geq  \lambda ^3 - O(N^{-1}).
\end{align}
Since on the other hand 
$$\tr \left[ Q\gamma_N ^{(1)} \right] = 1 -\lambda,$$
we conclude 
\begin{equation}\label{eq:control lambda}
	\tr\left[ \tH_3 ^R \widetilde{\gamma}_N ^{(3)} \right] + C \Lambda \tr[Q \gamma_N ^{(1)}] 
	\geq \lambda^3 \EAF_R + C \Lambda (1-\lambda)- C N ^{-1} 
	\geq \EAF_R - C N ^{-1}.
\end{equation}
For the last inequality we bound from below in terms of the infimum with 
respect to $0 \leq \lambda \leq 1$. 
Since $\Lambda$ is a very large number 
and $\EAF_R$ is bounded as $R \to 0$ (see Proposition~\ref{prop:af_limit}), 
the infimum is clearly attained at $\lambda = 1$.

\medskip

\noindent\textbf{Optimizing the error.} We next choose $\Lambda$ to minimize 
the error in~\eqref{eq:ener pre final}. 
We assume that $R$ behaves at worst as 
$$ R \sim N ^{-\eta}.$$
Changing a little bit $\eta$ if necessary we may ignore the $|\log R|$ and $R ^{\eps}$ factors, and we thus have to minimize
$$ \frac{1}{\sqrt{\Lambda} R} + \frac{1}{\Lambda R ^2} + \frac{\Lambda  ^{2+2/s}}{N}.$$
We pick 
$$ \Lambda  = N ^{\frac{2}{5} (1+\eta)(1+4/5s)^{-1}}$$
to equate the first and the last term and get 
$$ 
	\frac{1}{\sqrt{\Lambda} R} + \frac{1}{\Lambda R^2} + \frac{\Lambda^{2(1+1/s)}}{N} 
	= O(N ^{\frac{4}{5} (1+\eta)(1+1/s)(1+4/5s) ^{-1} - 1}) + O(N ^{\frac{8}{5} (1+\eta)(1+1/s)(1+4/5s) ^{-1} - 2}),
$$
and this is small provided 
$$\eta < \eta_0 := \frac{1}{4}\left( 1+ \frac{1}{s}\right) ^{-1}.$$
Since this is the main error term we conclude that the lower bound 
corresponding to~\eqref{eq:main ener} holds provided $R \sim N ^{-\eta}$ with $\eta < \eta_0$, as stated in the theorem. The limit $\EAF_R \to \EAF$ is dealt with in Appendix~\ref{app:AFF}.

\subsection{Convergence of states}\label{sec:state CV}

Given the previous constructions and energy estimates, the proof of~\eqref{eq:main state} follows almost exactly~\cite[Section 4.3]{LewNamRou-14} and is thus only sketched. 

Modulo extraction of subsequences we have 
$$\gamma_N ^{(k)} \wto_* \gamma ^{(k)}$$
weakly-$*$ in the trace-class as $N\to \infty$. 
From Proposition~\ref{pro:a priori} we know that 
$(-\Delta + V) \gamma_N ^{(1)}$ is uniformly bounded in trace-class. 
Under our assumptions, $(-\Delta + V) ^{-1}$ is compact and we may thus, 
modulo a further extraction, assume that 
$$ \gamma_N ^{(1)} \underset{N\to \infty}{\longrightarrow} \gamma ^{(1)}$$
strongly in trace-class norm. Then, by~\cite[Corollary 2.4]{LewNamRou-13}, 
we also have 
$$ \gamma_N ^{(k)} \underset{N\to \infty}{\longrightarrow} \gamma ^{(k)}$$
strongly for all $k\geq 0$.

Next we claim that the measure $\mu_N$ defined in 
Lemma~\ref{lem:deF-localized-state} converges (modulo extraction) 
to a limit probability measure $\mu\in \cP (S\gH)$ on the unit sphere of 
$\gH = L ^2 (\R ^2)$ and that 
\begin{equation}\label{eq:CVstate 0}
 \gamma ^{(k)} = \int_{u\in S\gH} |u ^{\otimes k}\rangle \langle u ^{\otimes k} | d\mu(u) \ \mbox{ for all } k\geq 0. 
\end{equation}
To see this, we first apply~\eqref{eq:deF estim}, 
with the above choice of $\Lambda$. We obtain
\begin{equation}\label{eq:CVstate 1} 
\Tr \left| P^{\otimes 3} \gamma_{N}^{(3)} P^{\otimes 3} - \int_{SP\gH} |u ^{\otimes 3}\rangle \langle u ^{\otimes 3} | d\mu_N(u) \right| \to 0. 
\end{equation}
On the other hand, combining~\eqref{eq:control lambda} with the energy upper 
bound \eqref{eq:ener upper bound} we get 
$$ \lambda \to 1$$
where $\lambda$ is the fraction of $P$-localized particles defined 
in~\eqref{eq:def lambda}. 
Using Jensen's inequality as in~\eqref{eq:lambda-lambda3} we deduce that 
\begin{equation}\label{eq:CVstate mass}
 \mu_N(SP\gH) =
 \tr\left[ P^{\otimes 3} \gamma_{N}^{(3)} P^{\otimes 3}\right] \to 1. 
\end{equation}
Combining with~\eqref{eq:CVstate 1} yields
\begin{equation}\label{eq:CVstate 2} 
\Tr \left| \gamma_{N}^{(3)} - \int_{SP\gH} |u ^{\otimes 3}\rangle \langle u ^{\otimes 3} | d\mu_N(u) \right| \to 0. 
\end{equation}
Testing this with a sequence of finite rank orthogonal projectors 
$$P_K \underset{K\to \infty}{\longrightarrow} \1$$
and using the strong convergence of $\gamma_N ^{(3)}$ gives
$$ \lim_{K\to \infty} \lim_{N\to \infty} \mu_N (P_K \gH)=1,$$
and we obtain the existence of a limit measure $\mu$ 
supported on the unit ball of $\gH$ by a tightness argument. 
Then~\eqref{eq:CVstate 0} for $k=3$ follows from~\eqref{eq:CVstate 2}. 
Since $\gamma_N ^{(3)}$ converges strongly, the limit has trace $1$ and $\mu$ 
must be supported on the unit sphere. 
Obtaining~\eqref{eq:CVstate 0} for larger $k$ is a general argument based 
on~\eqref{eq:CVstate mass}. 
We refer to~\cite[Section 4.3]{LewNamRou-14} for details.

There only remains to prove that $\mu$ is supported on $\cMAF$. 
But it follows easily from combining previously obtained energy bounds that 
$$ \int_{SP\gH} \left| \EAF_R - \cEAF_R [u] \right| d\mu_N (u) \underset{N\to \infty}{\longrightarrow} 0.$$
Using in addition the results of Appendix~\ref{app:AFF}, 
in particular Proposition~\ref{prop:af_limit}, 
we obtain for a large but fixed constant $C>0$
$$ \int_{\cEAF[u] \leq C} \left| \EAF - \cEAF [u] \right| d\mu_N (u) \underset{N\to \infty}{\longrightarrow} 0$$
and 
$$ \int_{\cEAF[u] \geq C} d\mu_N (u) \underset{N\to \infty}{\longrightarrow} 0.$$
Then clearly $\mu$ must be supported on $\cMAF$, which concludes the proof.\hfill\qed

\appendix

\section{Properties of the average-field functional} \label{app:AFF}

In this appendix we etablish some of the fundamental properties of the
functional \eqref{eq:avg func} and its limit $R \to 0$.

For $\beta \in \R$ and $V: \R^2 \to \R^+$ we define 
the average-field energy functional
\begin{equation}\label{eq:AFF}
	\cEAF[u] := \int_{\R^2} \left( \left| \left( \nabla + i \beta \bA[|u|^2] \right) u \right|^2 + V|u|^2 \right),
\end{equation}
with the self-generated magnetic potential
$$
	\bA[\rho] := \nablap w_0 * \rho = \int_{\R^2} \frac{(x-y)^\perp}{|x-y|^2} \rho(y) \,dy,
	\qquad \curl\:\bA[\rho] = 2\pi\rho.
$$
The functional is certainly well-defined for $u \in C_c^\infty(\R^2)$,
but we should ask what its natural domain is.
We then have to make a meaning of $\cEAF[u]$ for general $u \in L^2(\R^2)$
and the problem is that it is not certain that $\bA[|u|^2] \in L^2_{\loc}$
even though $u \in L^2$ (see \cite{ErdVou-02} for an example\footnote{Note 
that by Young's inequality we have for any $u \in L^2(\R^2)$ that
$\bA[|u|^2] \in L^p(\R^2) + \eps L^\infty(\R^2)$ for $p \in [1,2)$.
Also compare to the singular magnetic fields considered in 
\cite{ErdVou-02,LunSol-14}.}), 
so the product $\bA[|u|^2]u$ may not be well-defined as a distribution
(while $\nabla u$ certainly is).
One way around this is to reconsider the form of the functional when
acting on regular enough functions such that we can write 
$u = |u|e^{i\varphi}$ where $\varphi$ is real. Then
$$
	\left| (\nabla + i\beta\bA[|u|^2])u \right|^2
	= \left| \nabla|u| + i|u|(\nabla\varphi + \beta\bA[|u|^2]) \right|^2
	= \left| \nabla|u| \right|^2 
		+ \left| |u|\nabla\varphi + \beta \bA[|u|^2]|u| \right|^2,
$$
where also $\nabla \varphi = |u|^{-2}\Im \bar{u}\nabla u$
and $\nabla |u| = |u|^{-1}\Re \bar{u}\nabla u$.
Hence, an alternative definition is given by
\begin{equation}\label{eq:AFF_alt}
	\cEAF[u] := \int_{\R^2} \left( \left| \nabla|u| \right|^2 
		+ \left| \Im \frac{\bar{u}}{|u|}\nabla u + \beta \bA[|u|^2]|u| \right|^2
		+ V|u|^2 \right),
\end{equation}
and the advantage of this formulation is that it makes clear that we actually
demand $|u| \in H^1(\R^2)$ in order for $\cEAF[u] < \infty$.
We can then use the following lemma to see that in fact 
$\bA[|u|^2]u \in L^2(\R^2)$,
and hence also $\nabla u \in L^2(\R^2)$.
(And conversely this also shows that if $\bA[|u|^2]u \notin L^2(\R^2)$ 
then we have no chance of making sense out of $\cEAF[u]$.)

\begin{lemma}[\textbf{Bound on the magnetic term}]\label{lem:af_three_body}\mbox{}\\
	We have for any $u \in L^2(\R^2)$ that
	$$
		\int_{\R^2} \left| \bA[|u|^2] \right|^2 |u|^2 
		\le \frac{3}{2} \|u\|_{L^2(\R^2)}^4 \int_{\R^2} \left| \nabla |u| \right|^2.
	$$
\end{lemma}
\begin{proof}
	This follows from symmetry and 
	from the three-body Hardy inequality of Lemma~\ref{lem:Hardy}:
	\begin{align*}
		& \int_{\R^2} \left| \bA[|u|^2](x) \right|^2 |u(x)|^2 \,dx 
		= \iiint_{\R^6} \frac{x-y}{|x-y|^2} \cdot \frac{x-z}{|x-z|^2} |u(x)|^2 |u(y)|^2 |u(z)|^2 \,dxdydz \\
		&= \frac{1}{6} \int_{\R^6} \frac{1}{\cR(X)^2} \left| |u|^{\otimes 3} \right|^2 dX 
		\le \frac{1}{2} \int_{\R^6} \left| \nabla_X |u|^{\otimes 3} \right|^2 dX
		= \frac{3}{2} \int_{\R^2} \left| \nabla|u| \right|^2 dx
		\left( \int_{\R^2} |u|^2 dx \right)^2.
	\end{align*}
\end{proof}

We can therefore define the domain of $\cEAF$ to be
(and otherwise let $\cEAF[u] := +\infty$)
$$
	\cDAF := \left\{ u \in H^1(\R^2) : \int_{\R^2} V|u|^2 < \infty \right\},
$$
and we find using Cauchy-Schwarz, Lemma~\ref{lem:af_three_body},
and $|\nabla|u|| \le |\nabla u|$ that for $u \in \cDAF$
$$
	0 \le \cEAF[u] 
	\le 2 \|\nabla u\|^2 + 2\beta^2 \| \bA[|u|^2] u \|^2 + \int V|u|^2
	\le (2 + 3\beta^2 \|u\|^4) \|\nabla u\|^2 + \int V|u|^2 < \infty.
$$
The ground-state energy of the average-field functional is then given by
$$
	\EAF := \inf \left\{ \cEAF[u] : u \in \cDAF, \int_{\R^2} |u|^2 = 1 \right\}.
$$
For convenience we also make the assumption on $V$ that 
$V(x) \to +\infty$ as $|x| \to \infty$
and that $C_c^\infty(\R^2) \subseteq \cDAF$ is a form core
for $\norm{u}_{L^2_V}^2 := \int_{\R^2} V|u|^2$,
with $-\Delta + V$ essentially self-adjoint 
and with purely discrete spectrum (see, e.g., \cite[Theorem XIII.67]{ReeSim4}).
This is then also a core for $\cEAF$:

\begin{proposition}[\textbf{Density of regular functions in the form domain}]\label{prop:af_density}\mbox{}\\
	$C_c^\infty(\R^2)$ is dense in $\cDAF$ w.r.t. $\cEAF$,
	namely for any $u \in \cDAF$ 
	there exists a sequence 
	$(u_n)_{n \to \infty} \subset C_c^\infty(\R^2)$ such that
	$$\|u-u_n\|_{H^1} \to 0
	\mbox{ and }\cEAF[u_n] \to \cEAF[u] \mbox{ as }n \to \infty.$$
\end{proposition}
\begin{proof}
	Take $u \in \cDAF$, then $\|\nabla u\|_{L^2} < \infty$ and hence also 
	$\|u\|_{L^p} < \infty$ for any $p \in [2,\infty)$ by Sobolev embedding.
	We use that $C_c^\infty(\R^2)$ is dense in $H^1(\R^2)$,
	so there exists a sequence $(u_n)_{n \to \infty} \subset C_c^\infty$ 
	s.t. $\|u-u_n\|_{H^1} \to 0$. Also,
	\begin{align*}
		&\left| \norm{ (\nabla + i\beta\bA[|u|^2])u }_2 - \norm{ (\nabla + i\beta\bA[|u_n|^2])u_n }_2 \right| \\
		&\le \norm{ (\nabla + i\beta\bA[|u|^2])u - (\nabla + i\beta\bA[|u_n|^2])u_n }_2 \\
		&\le \norm{ \nabla(u-u_n) }_2 
			+ |\beta| \| (\bA[|u|^2] - \bA[|u_n|^2])u + \bA[|u_n|^2](u-u_n) \|_2 \\
		&\le \norm{u-u_n}_{H^1} + |\beta| \norm{ \bA[|u|^2-|u_n|^2] u }_2 
			+ |\beta|\norm{ \bA[|u_n|^2](u-u_n) }_2,
	\end{align*}
	where by H\"older's and generalized Young's inequalities
	\begin{align*}
		&\norm{\bA[|u|^2-|u_n|^2] u}_2 
		\le \norm{\bA[|u|^2-|u_n|^2]}_{4} \norm{u}_{4}
		\le C \norm{|u|^2-|u_n|^2}_{4/3} \norm{\nabla w_0}_{2,w} \norm{u}_4 \\
		&\le C' \norm{u-u_n}_{8/3}
		\le C'' \norm{u-u_n}_{H^1} \to 0,
	\end{align*}
	and similarly 
	$$\norm{\bA[|u_n|^2](u-u_n)}_2 \le C\norm{u-u_n}_{H^1} \to 0,$$
	as $n \to \infty$.

	We also have continuity for $\norm{u}_{L^2_V}	$ here since we assumed that
	$C_c^\infty(\R^2)$ is a form core.
\end{proof}

\begin{lemma}[\textbf{Basic magnetic inequalities}]\label{lem:af_smooth_ineqs}\mbox{}\\
	We have for $u \in \cDAF$ that (diamagnetic inequality)
	\begin{equation}\label{eq:af_diamagnetic}
		\int_{\R^2} \left| (\nabla + i\beta \bA[|u|^2])u \right|^2 
		\ge \int_{\R^2} \left| \nabla |u| \right|^2,
	\end{equation}
	and
	\begin{equation}\label{eq:af_lower_bound}
		\int_{\R^2} \left| (\nabla + i\beta \bA[|u|^2])u \right|^2 
		\ge 2\pi|\beta| \int_{\R^2} |u|^4.
	\end{equation}
\end{lemma}
\begin{proof}
	By density we can w.l.o.g. assume  $u \in C_c^\infty(\R^2)$.
	We then have $\bA[|u|^2] \in C^\infty(\R^2) \subseteq L^2_{\loc}(\R^2)$ 
	and hence the first inequality follows by the usual diamagnetic inequality
	(see e.g. Theorem~2.1.1 in \cite{Fournais-Helffer}).
	Furthermore, by e.g. Lemma~1.4.1 in \cite{Fournais-Helffer},
	$$
		\int_{\R^2} \left| (\nabla + i\beta \bA[|u|^2])u \right|^2 
		\ge \pm \int_{\R^2} \curl\left( \beta\bA[|u|^2] \right) |u|^2,
	$$
	which proves the second inequality since $\curl \bA[|u|^2] = 2\pi|u|^2$.
	Instead of using density we could also have used the formulation
	\eqref{eq:AFF_alt} or the fact that 
	$u \in H^1 \Rightarrow \bA[|u|^2] \in L^p$, 
	$p \in (2,\infty)$ by generalized Young.
\end{proof}

\begin{proposition}[\textbf{Existence of minimizers}]\label{prop:af_minimizer}\mbox{}\\
	For any value of $\beta \in \R$
	there exists $\uAF \in \cDAF$ with $\int_{\R^2} |\uAF|^2 = 1$ and
	$\cEAF[\uAF] = \EAF$.
\end{proposition}
\begin{proof}
	First note that for $u \in \cDAF$, by Lemma~\ref{lem:af_three_body}
	and Lemma~\ref{lem:af_smooth_ineqs},
	\begin{align*}
		\norm{\nabla u}_2 &= \norm{\nabla u + i\beta \bA[|u|^2] u - i\beta \bA[|u|^2] u} _2
		\le \cEAF[u]^{1/2} + |\beta| \norm{\bA[|u|^2]u}_2 \\
		&\le \cEAF[u]^{1/2} + |\beta| \sqrt{\frac{3}{2}} \norm{u}_2^2 \norm{\nabla|u|}_2 
		\le \left( 1 + |\beta| \sqrt{\frac{3}{2}} \norm{u}_2^2 \right) \cEAF[u]^{1/2}.
	\end{align*}
	Now take a minimizing sequence 
	$$(u_n)_{n \to \infty} \subset \cDAF,\:
	\norm{u_n}_2 = 1,\: 
	\lim_{n \to \infty} \cEAF[u_n] = \EAF.$$
	Then clearly $(u_n)$ is uniformly bounded in both $L^2(\R^2)$,
	$L^2_V$, and $H^1(\R^2)$ (and hence in $L^p(\R^2)$, $p \in [2,\infty)$),
	and therefore by the Banach-Alaoglu theorem there exists $\uAF \in \cDAF$ 
	and a weakly convergent subsequence (still denoted $u_n$) such that
	$$
		u_n \rightharpoonup \uAF \ \text{in} \ L^2(\R^2) \cap L^2_V \cap L^p(\R^2), \quad
		\nabla u_n \rightharpoonup \nabla \uAF \ \text{in} \ L^2(\R^2).
	$$
	Moreover, since $(-\Delta + V + 1)^{-1/2}$ is compact we have that
	$$u_n = (-\Delta+V + 1)^{-1/2}(-\Delta+V + 1)^{1/2}u_n$$ 
	is actually strongly convergent (again extracting a subsequence), 
	hence 
	$$u_n \to \uAF \mbox{ in } L^2(\R^2).$$
	Also, $\bA[|u_n|]$ converges pointwise a.e. to $\bA[|u|^2]$ by weak
	convergence of $u_n$ in $L^p$
	and, by the trick of Lemma~\ref{lem:af_three_body},
	$$
		\norm{\bA[|u_n|^2]u_n}_2^2
		= \frac{1}{6} \int_{\R^6} \cR(X)^{-2} \left| |u_n|^{\otimes 3} \right|^2 dX 
		\to \frac{1}{6} \int_{\R^6} \cR(X)^{-2} \left| |u|^{\otimes 3} \right|^2 dX 
		= \norm{\bA[|u|^2]u}_2^2
	$$
	by dominated convergence.
	The functions $\bA[|u_n|^2]u_n$ are therefore even strongly converging 
	to $\bA[|u|^2]u$ in $L^2(\R^2)$ by dominated convergence.
	It then follows that
	\begin{align*}
		\norm{ (\nabla + i\beta \bA[|u|^2])u }_2 
		&= \sup_{\|v\|=1} |\langle \nabla u + i\beta \bA[|u|^2]u, v \rangle| \\
		&= \sup_{\|v\|=1} \lim_{n \to \infty} |\langle \nabla u_n + i\beta \bA[|u_n|^2]u_n, v \rangle| \\
		&\le \liminf_{n \to \infty} \sup_{\|v\|=1} |\langle \nabla u_n + i\beta \bA[|u_n|^2]u_n, v \rangle| \\
		&= \liminf_{n \to \infty} \norm{ (\nabla + i\beta \bA[|u_n|^2])u_n }_2,
	\end{align*}
	and since $\norm{\cdot}_{L^2_V}$ is also weakly lower semicontinuous
	(see, e.g., \cite[Supplement to IV.5]{ReeSim1}),
	we have $\liminf_{n \to \infty} \cEAF[u_n] \ge \cEAF[\uAF]$.
	Thus, with $\|\uAF\| = \lim_{n \to \infty} \|u_n\| = 1$,
	we also have $\cEAF[\uAF] = \EAF$.
\end{proof}

\begin{proposition}[\textbf{Convergence to bosons}] \label{prop:af_bosons}\mbox{}\\
	Let $E_0$ resp. $u_0$ denote the ground-state eigenvalue resp.
	normalized eigenfunction of the non-magnetic
	Schr\"odinger operator $H_1 = -\Delta + V$,
	with $V \in L^\infty_\loc$.
	We have 
	$$\EAF_{(\beta)} \underset{\beta\to 0}{\to} E_0,$$
	and that given an arbitrary sequence $(u_\beta)$
	of minimizers for $\cEAF_{(\beta)}$
	$$u_\beta \underset{\beta\to 0}{\to} u_0 \mbox{ in } L^2(\R^2)$$
	up to a subsequence and a constant phase.
\end{proposition}
\begin{proof}
	Note that under our conditions for $V$,
	$u_0 \in \cDAF$ is the unique minimizer of 
	$\cE_0 = \cEAF_{(\beta=0)}$ and can be taken positive
	(see, e.g., \cite[Theorem 11.8]{LieLos-01}).
	By the diamagnetic inequality \eqref{eq:af_diamagnetic},
	and by taking the trial state $u_0 = |u_0|$ in $\cEAF_{(\beta \neq 0)}$,
	we find
	$$
		E_0 \le \EAF_{(\beta)} \le \cEAF_{(\beta)}[u_0] 
		= \cE_0[u_0] + \beta^2\norm{\bA[|u_0|^2]u_0}_2^2
		\le (1 + C\beta^2)E_0
	$$
	(where we also used 
	Lemma~\ref{lem:af_three_body}),
	and hence $\EAF_{(\beta)} \to E_0$ as $\beta \to 0$.
	Now consider a sequence $(u_\beta) \subset \cDAF$ 
	of minimizers as $\beta \to 0$
	with $\cEAF[u_\beta] \to E_0$, $\|u_\beta\|=1$. 
	Then, because of uniform boundedness and as in the proof of
	Proposition~\ref{prop:af_minimizer}, we have after taking a subsequence
	that $u_\beta \to u$ for some $u \in \cDAF$, $\|u\|=1$, and also
	\begin{align*}
		\|\nabla u\| &= \sup_{\|v\|=1} \left| \langle \nabla u,v \rangle \right|\\
		&= \sup_{\|v\|=1} \lim_{\beta \to 0} \left| \langle \nabla u_\beta 
			+ i\beta\bA[|u_\beta|^2]u_\beta,v \rangle \right|\\
		&\le \liminf_{\beta \to 0} \norm{ \nabla u_\beta 
			+ i\beta\bA[|u_\beta|^2]u_\beta },
	\end{align*}
	so 
	$$E_0 \le \cE_0[u] \le \liminf_{\beta \to 0} \cEAF_{(\beta)}[u_\beta].$$
	It follows that $\cE_0[u] = E_0$ and hence $u=u_0$ 
	up to a constant phase.
\end{proof}

From the bound \eqref{eq:af_lower_bound} we observe that the self-generated 
magnetic interaction is stronger than a contact interaction of strength 
$2\pi|\beta|$
(despite the fact that we already removed a singular repulsive interaction in the initial
regularization step for extended anyons). 
Hence we have not only $\EAF \ge E_0$ by the diamagnetic inequality, but also
\begin{equation}\label{eq:af_contact_bound}
	\EAF \ge \min_{\rho \ge 0,\ \int_{\R^2} \rho = 1}
		\int_{\R^2} \left( 2\pi|\beta|\rho^2 + V\rho \right),
\end{equation}
which can be computed for given $V$ by straightforward optimization.

\medskip

Let us now consider the corresponding situation for the regularized functional
(extended anyons)
$$
	\cEAF_R[u] := \int_{\R^2} \left( \left| \left( \nabla + i \beta \bA^R[|u|^2] \right) u \right|^2 + V|u|^2 \right),
	\quad \bA^R[\rho] := \nablap w_R * \rho,
	\quad R > 0.
$$
Since $\nabla w_R \in L^\infty(\R^2)$ we have 
$\bA^R[|u|^2] \in L^\infty(\R^2)$ 
with 
$$\norm{\bA^R[|u|^2]}_\infty \le \frac{C}{R}\|u\|_2^2$$
and instead of Lemma~\ref{lem:af_three_body} we have
$$\norm{\bA^R[|u|^2]u}_2 \le C\|u\|_2^2 \||u|\|_{H^1}$$ 
using Lemma~\ref{lem:three body}.
Hence the natural domain is again $\cDAF$
and all properties established above 
for $\cEAF$ are also found to be valid for $\cEAF_R$ 
(except \eqref{eq:af_lower_bound} and \eqref{eq:af_contact_bound}
which now have regularized versions).
Denoting 
$$\EAF_R := \min \{\cEAF_R[u] : u \in \cDAF, \|u\|_2 = 1\},$$
we furthermore have the following relationship:

\begin{proposition}[\textbf{Convergence to point-like anyons}] \label{prop:af_limit}\mbox{}\\
	The functional $\cEAF_R$ converges pointwise to $\cEAF$ as $R \to 0$. More precisely, for any $u\in \cDAF$
	\begin{equation}\label{eq:ext to point}
	 \left| \cEAF_R [u] - \cEAF [u] \right| \leq 
	C_u |\beta|(1+\beta^4) (1+\cEAF[u])^{3/2} R,
	\end{equation}
	where $C_u$ depends only on $\norm{u}_2$. 
	Hence, 
	$$\EAF_R \underset{R\to 0}{\to} \EAF,$$  
	and if $(u_R)_{R \to 0} \subset \cDAF$ denotes 
	a sequence of minimizers of $\cEAF_R$, then there exists a
	subsequence $(u_{R'})_{R' \to 0}$ s.t. $u_{R'} \to \uAF$ as $R' \to 0$,
	where $\uAF$ is some minimizer of $\cEAF$.
\end{proposition}
\begin{proof}
	We have for any $u \in \cDAF$ that
	\begin{align*}
		&\left| \norm{ (\nabla + i\beta\bA[|u|^2])u }_2 - \norm{ (\nabla + i\beta\bA^R[|u|^2])u }_2 \right| \\
		&\le \norm{ (\nabla + i\beta\bA[|u|^2])u - (\nabla + i\beta\bA^R[|u|^2])u }_2 
		= |\beta| \norm{ (\bA[|u|^2] - \bA^R[|u|^2])u }_2 \\
		&\le |\beta| \norm{ \bA[|u|^2] - \bA^R[|u|^2] }_4 \norm{u}_4
		= |\beta| \norm{ (\nabla w_0 - \nabla w_R) * |u|^2 }_4 \norm{u}_4,
	\end{align*}
	where by Young
	$$
		\norm{ (\nabla w_0 - \nabla w_R) * |u|^2 }_4 
		\le \norm{ \nabla w_0 - \nabla w_R }_1 \norm{|u|^2}_4
		\le \norm{\nabla w_0}_{L^1(B(0,R))} \norm{u}_8^2
		\to 0,
	$$
	as $R \to 0$, since $\nabla w_0 \in L^1_\loc(\R^2)$. 
	We deduce~\eqref{eq:ext to point} by combining this with previous 
	estimates of this appendix and Sobolev embeddings. 
	It follows that $\cEAF_R[u] \to \cEAF[u]$ as $R \to 0$.
	
	Let $(u_R)_{R \to 0}$ denote a sequence of minimizers of $\cEAF_R$:
	$$\EAF_R = \cEAF_R[u_R],\: \|u_R\| = 1,$$
	and take $u \in \cDAF$ an arbitrary minimizer of $\cEAF$.
	Then, since 
	$$\EAF_R \le \cEAF_R[u] \underset{R\to 0}{\to} \cEAF[u] = \EAF,$$
	we have that $\EAF_R$ is uniformly bounded as $R \to 0$ and that
	$$\limsup_{R \to 0} \EAF_R \le \EAF.$$
	Then $\cEAF_R[u_R]$, and hence also
	$$
		\cEAF[u_R] \le C\left( \|u_R\|_{H^1}^2 + \|u_R\|_{L^2_V}^2 \right) 
		\le C'(\cEAF_R[u_R] + 1),
	$$
	are uniformly bounded as well.
	As in the proof of Proposition~\ref{prop:af_minimizer},
	there then exists a strongly convergent subsequence $(u_{R'})_{R' \to 0}$,
	with $u_{R'} \to u_0 \in \cDAF$.
	Also, by weak lower semicontinuity
	$\cEAF[u_0] \le \liminf_{R' \to 0} \cEAF[u_{R'}]$, so that
	for any $\eps > 0$ and sufficiently small $R'>0$,
	$$
		\EAF \le \cEAF[u_0] 
		\le \cEAF[u_{R'}] + \eps \le \cEAF_{R'}[u_{R'}] + 2\eps
		= \EAF_{R'} + 2\eps,
	$$
	where we also used that the convergence is uniform for our 
	uniformly bounded sequence $u_R$ by the bound \eqref{eq:ext to point}.
	It follows that $\EAF \le \cEAF[u_0] \le \EAF + 3\eps$,
	and hence $u_0$ is a minimizer with $\|u_0\|=1$ and
	$\EAF = \cEAF[u_0] = \lim_{R \to 0} \EAF_R$.
\end{proof}


\end{document}